%% file: main.tex
\PassOptionsToPackage{usenames,dvipsnames}{xcolor}
\documentclass[letterpaper,11pt]{article}
\usepackage[nottoc]{tocbibind}

\usepackage[utf8x]{inputenc}
\usepackage[T1]{fontenc}
\usepackage[normalem]{ulem}
\usepackage{amsmath}

\usepackage[margin=1in]{geometry}
\usepackage{graphicx}
\usepackage{tikz}
\usepackage{comment}
\usepackage{enumitem}
\makeatletter
\DeclareFontFamily{U}{tipa}{}
\DeclareFontShape{U}{tipa}{m}{n}{<->tipa10}{}
\newcommand{\arc@char}{{\usefont{U}{tipa}{m}{n}\symbol{62}}}%

\newcommand{\arc}[1]{\mathpalette\arc@arc{#1}}

\newcommand{\R}{\mathbb{R}}

\newcommand{\arc@arc}[2]{%
  \sbox0{$\m@th#1#2$}%
  \vbox{
    \hbox{\resizebox{\wd0}{\height}{\arc@char}}
    \nointerlineskip
    \box0
  }%
}
\makeatother

\usepackage[usenames,dvipsnames]{xcolor}

    \definecolor{DarkRed}{rgb}{0.368,0.097,0.078}
\definecolor{DarkBlue}{rgb}{0.2,0.2,0.6}

\usepackage[colorlinks = true,
    linkcolor = DarkBlue,
    anchorcolor = DarkBlue,
    citecolor = DarkBlue,
    filecolor = DarkBlue,
    urlcolor = DarkRed,
    pagebackref]{hyperref}

\newcommand{\eps}{\varepsilon}

\input{header.tex}

\title{Smoothed Analysis of Nash Equilibrium and Fixed Point Computation, and Efficiency of the Homotopy Method}

\author{\parbox{\textwidth}{\centering
Idan Attias$^{1}$\quad
Yuval Dagan$^{2}$\quad
Constantinos Daskalakis$^{3}$\\[0.4em]
Rui Yao$^{3}$%
\thanks{Main contribution. The authors are listed in alphabetical order.}\quad
Manolis Zampetakis$^{4}$\\[0.8em]
{\small
$^{1}$Institute for Data, Econometrics, Algorithms, and Learning,
UIC and TTIC\\
$^{2}$School of Computer Science, Tel Aviv University\\
$^{3}$EECS Department, MIT\\
$^{4}$Department of Computer Science, Yale University
}}}

\begin{document}
\pagenumbering{roman}
\maketitle
\begin{NoHyper}
\begingroup
\renewcommand{\thefootnote}{}
\footnotetext{%
Emails:
\texttt{idanattias88@gmail.com},
\texttt{ydagan@tauex.tau.ac.il},
\texttt{costis@csail.mit.edu},
\texttt{rayyao@mit.edu},
\texttt{emmanouil.zampetakis@yale.edu}
}
\endgroup
\end{NoHyper}
\thispagestyle{empty}

\begin{abstract}
We study, under a smoothed analysis perspective, the complexity of three PPAD-complete problems: computing Nash equilibria, computing fixed points, and solving variational inequalities. We consider a general perturbation model that bounds the probability density of the quantities evaluated by our algorithms. As special cases, it captures worst-case games perturbed by small random per-action costs, and worst-case maps perturbed by a small random affine function.

Under these models, we show that classical algorithmic paradigms, cover-based search and homotopy path-following provably outperform the worst-case state-of-the-art. For games with a constant number of players and $n$ strategies per player, a recursive cover-based algorithm computes an $\eps$-Nash equilibrium using $\log(1/\eps)\cdot 2^{O(n)}$ expected payoff-vector computations, compared to the worst-case runtime tradeoffs of $\log(1/\eps)\cdot n^{O(n)}$, $(1/\eps)^{O(n)}$, $2^{O(\log^2 n/\eps^2)}$, and $2^{O(n)}\mathrm{poly}(\log(1/\eps))$. For fixed points and variational inequalities of bounded and Lipschitz maps on the $n$-dimensional $\ell_2$ unit ball, the same approach runs in expected time $\log(1/\eps)\cdot 2^{O(n)}$, compared to the worst-case $(1/\eps)^{O(n)}$; the $2^{\Omega(n)}$ factor is necessary by prior work~\cite{rubinstein2017settling}. For general smooth maps, we further analyze the classical predictor--corrector homotopy method, for which only qualitative convergence guarantees were previously known, showing that under perturbations the homotopy path is well conditioned and has length $2^{O(n)}$, so that the method finds an $\eps$-approximate solution in time $2^{O(n)}\cdot\mathrm{poly}(1/p)/\eps$ with probability at least $1-p$.
\end{abstract}


\newpage
\tableofcontents
\newpage

\pagenumbering{arabic}

\input{intro}

\input{nash}

\input{rec-fixed}

\input{dynamics}

\input{discrete}

\input{assumption}

\input{smooth}

\input{lowerBound}

\section*{AI disclosure}

We used OpenAI ChatGPT to assist with checking and revising mathematical arguments in Sections 2--4, and with proofreading the manuscript. The authors verified the correctness and originality of all content, including references.

\bibliographystyle{alpha}
\bibliography{bib}

\end{document}

%% file: header.tex
\usepackage{amsthm} 
\usepackage{mathtools,thmtools}

\usepackage{enumitem}
\usepackage{bm}
\usepackage{xspace}
\usepackage{nicefrac}

\usepackage[capitalise]{cleveref}
\usepackage{dsfont}

\usepackage[usenames,dvipsnames]{xcolor}

\declaretheoremstyle[
	    spaceabove=\topsep, 
	    spacebelow=\topsep, 
	    headfont=\normalfont\bfseries,
	    bodyfont=\normalfont\itshape,
	    notefont=\normalfont\bfseries,
	    notebraces={(}{)},
	    postheadspace=0.5em, 
	    headpunct={},
    ]{theorem}
\declaretheorem[style=theorem,numberwithin=section]{theorem}

\declaretheoremstyle[
	    spaceabove=\topsep, 
	    spacebelow=\topsep, 
	    headfont=\normalfont\bfseries,
	    bodyfont=\normalfont,
	    notefont=\normalfont\bfseries,
	    notebraces={(}{)},
	    postheadspace=0.5em, 
	    headpunct={},
    ]{definition}

\declaretheoremstyle[
        spaceabove=\topsep, 
        spacebelow=\topsep, 
        headfont=\normalfont\bfseries,
        bodyfont=\normalfont,
        notefont=\normalfont\bfseries,
        notebraces={}{},
        postheadspace=0.5em, 
        qed=$\blacksquare$, 
        headpunct={},
    ]{proofstyle}
\declaretheorem[style=proofstyle,numbered=no,name=Proof]{proof}

\declaretheorem[style=theorem,sibling=theorem,name=Lemma]{lemma}
\declaretheorem[style=theorem,sibling=theorem,name=Assumption]{assumption}

\declaretheorem[style=theorem,sibling=theorem,name=Corollary]{corollary}

\declaretheorem[style=theorem,sibling=theorem,name=Informal Theorem]{informal}

\declaretheorem[style=theorem,sibling=theorem,name=Property]{property}

\declaretheorem[style=theorem,numbered=no,name=Theorem]{theorem*}
\declaretheorem[style=theorem,numbered=no,name=Lemma]{lemma*}
\declaretheorem[style=theorem,numbered=no,name=Corollary]{corollary*}
\declaretheorem[style=theorem,numbered=no,name=Proposition]{proposition*}
\declaretheorem[style=theorem,numbered=no,name=Claim]{claim*}
\declaretheorem[style=theorem,numbered=no,name=Fact]{fact*}
\declaretheorem[style=theorem,numbered=no,name=Observation]{observation*}
\declaretheorem[style=theorem,numbered=no,name=Conjecture]{conjecture*}

\declaretheorem[style=definition,sibling=theorem,name=Definition]{definition}
\declaretheorem[style=definition,sibling=theorem,name=Remark]{remark}
\declaretheorem[style=definition,sibling=theorem,name=Example]{example}
\declaretheorem[style=definition,sibling=theorem,name=Question]{question}

\declaretheorem[style=definition,numbered=no,name=Definition]{definition*}
\declaretheorem[style=definition,numbered=no,name=Remark]{remark*}
\declaretheorem[style=definition,numbered=no,name=Example]{example*}
\declaretheorem[style=definition,numbered=no,name=Question]{question*}

\DeclareMathAlphabet{\mathbfsf}{\encodingdefault}{\sfdefault}{bx}{n}

\let\Pr\relax
\DeclareMathOperator{\Pr}{\mathbb{P}}

\newcommand{\lrnorm}[1]{\mathopen{}\left\|#1\right\|}

\usepackage{amsfonts,amsmath,amssymb}
\usepackage{mathtools}
\usepackage{float}
\usepackage{enumitem}
\usepackage{algorithm}
\usepackage{algpseudocode}

\newcommand{\cB}{\mathcal{B}}

\newcommand{\cN}{\mathcal{N}}

\newcommand{\cC}{\mathcal{C}}

\newcommand{\hR}{\mathbb{R}}

\newcommand{\hS}{\mathbb{S}}



%% file: intro.tex
\section{Introduction}

Brouwer's fixed point theorem~\cite{Brouwer1911} states that a continuous mapping, $G$, from a non-empty, convex and compact subset, $S$, of $\mathbb{R}^n$ to itself must have a fixed point, i.e., some point $x \in S$ such that $G(x)=x$. It is a celebrated result in topology that has found applications across a wide range of theoretical and applied fields, including Combinatorics, Optimization, Game Theory, Economics, and Social Choice Theory. It was famously used by John Nash~\cite{nash1950equilibrium,nash1951brouwer} to establish the existence of mixed Nash equilibria in finite one-shot games, where each player chooses a distribution over a finite set of actions to optimize the expected value of an objective that depends on their own as well as the other players' actions. This jump-started a slew of developments that over the past 75 years solidified the foundations of Game Theory and neoclassical Economics; see, e.g.,~\cite{myerson1999nash}. 

From a computational standpoint, finding Brouwer fixed points is challenging. Many classical iterative methods have been proposed, but they typically offer only asymptotic convergence guarantees without explicit finite-time complexity bounds (or suffer from worst-case query and time complexities that scale exponentially with the dimension), unless the problem possesses special structure---for example, when $G$ is a contractive mapping with a Lipschitz constant bounded away from $1$, in which case geometric convergence to the unique fixed point of $G$ is guaranteed. In the general case of Lipschitz functions, a standard computational target is finding {\em approximate fixed points}, i.e., points satisfying $\|G(x)-x\|\le \varepsilon$ for a desired accuracy $\varepsilon > 0$. This problem is PPAD-complete~\cite{papadimitriou1994complexity} when white-box access to $G$ is provided, and several exponential-in-the-dimension query complexity lower bounds have also been established when $G$ is accessed via black-box value queries, starting with the seminal work of Hirsch, Papadimitriou, and Vavasis~\cite{hirsch1989exponential}; see our related work section for further details.

In several settings, we have a sharp understanding of the best achievable upper and lower bounds on the query complexity of computing approximate fixed points. These bounds depend on the domain, the norm with respect to which the Lipschitzness of $G$ is measured, and the norm used to measure the approximation error. For example, when the domain is $S=[0,1]^n$, $G$ is $L_1$-Lipschitz with respect to $\ell_\infty$, and we seek $\varepsilon$-approximate fixed points in the $\ell_\infty$ norm, the best known upper bound on the query complexity is $O(L_1/\varepsilon)^{n-1}$, which is near-tight~\cite{chen2008matching}. Similar bounds are known for other domains and norms.

A primary motivation for our work is that, in many black-box computation settings, the asymptotic query and time complexities of general-purpose approximate fixed-point algorithms are not qualitatively better than those of naive brute-force search---which simply enumerates over a cover of the domain by sufficiently small $\ell_p$-balls to identify a point that is an approximate fixed point. Indeed, in worst-case instances, qualitative improvements over brute-force search are precluded by existing lower bounds.

A natural way to go beyond these worst-case barriers is \emph{smoothed analysis}~\cite{spielman2004smoothed}, which interpolates between worst-case and average-case analysis: an adversary chooses an instance, which is then randomly perturbed, and one bounds the running time in expectation (or with high probability) over the perturbation. Smoothed analysis has explained the practical success of many algorithms with exponential worst-case behavior, starting with the simplex method~\cite{spielman2004smoothed}; see~\cite{spielman2009smoothed,roughgarden2021beyond} for surveys. Within this context, a central question motivating our work is the following:

\begin{question}
Can we leverage a smoothed analysis perspective to obtain algorithms that improve upon the worst-case state-of-the-art, in terms of query complexity and running time, for approximate fixed point computation?
\end{question}

Beyond designing new algorithms, we would also like to understand how classical algorithms for fixed-point computation behave under smoothed analysis. Inspired by the combinatorial nature of the Lemke--Howson algorithm~\cite{lemke1964equilibrium} for Nash equilibrium computation in two-player games and Lemke's algorithm~\cite{lemke1965bimatrix} for linear complementarity problems, Scarf~\cite{scarf1967approximation} proposed a seminal algorithm for computing approximate fixed points of general continuous functions, which Kuhn~\cite{kuhn1968simplicial} later recast in terms of simplicial subdivisions and Sperner's lemma~\cite{sperner1928neuer}. This work launched a rich literature on simplicial and homotopy-based path-following methods~\cite{eaves1972homotopies,merrill1972applications,van1982computation,todd2013computation}. Because these general-purpose algorithms access $G$ only through evaluations, they cannot escape the aforementioned exponential lower bounds; nevertheless, they remain widely used in practice when the dimension is small. When the map has additional structure, such as smoothness, a classical alternative is the \emph{predictor--corrector} homotopy method~\cite{kellogg1976constructive,smale1976convergent,chow1978finding,garcia1981pathways,allgower2003introduction}. It continuously deforms a simple function $h$, whose fixed point is known, into the target function $G$ through a family of intermediate functions $G_t(x)=(1-t)\,h(x)+t\,G(x)$, $t\in[0,1]$. The fixed points of these intermediate functions form a path that starts at the known fixed point of $h$ and leads to a fixed point of $G$, and the method follows this path by alternating tangent (predictor) steps with Newton (corrector) steps that use first-order (Jacobian) information. Variants of this method underlie standard numerical software~\cite{watson1987hompack} and global Newton methods for computing economic and game-theoretic equilibria~\cite{smale1976convergent,govindan2003global,herings2010homotopy}, and are closely related to the ridge-following dynamics of~\cite{daskalakis2023stay} for min-max critical points of nonconvex-nonconcave games. For systems of polynomial equations, a celebrated line of work established polynomial average-case and smoothed bounds for homotopy continuation~\cite{shub1993complexity,shub1994complexity,beltran2011fast,burgisser2011problem,lairez2017deterministic}. For general smooth maps, however, the guarantees for these methods are qualitative: in the worst case, the homotopy path may be nearly singular or exponentially long, and no quantitative running-time bounds are known (see the related work section for further details). This leads to our second question.

\begin{question}\label{q:paradigms}
Do standard algorithmic paradigms for fixed-point computation (cover-based search and homotopy path-following) provably achieve better running times under smoothed analysis?
\end{question}

We next turn to the other paradigmatic PPAD-complete problem~\cite{DaskalakisGP06,chen2006computing}: computing Nash equilibria in normal-form games (see, e.g.,~\cite{lemke1964equilibrium,savani2006hard,AdsulGMSS21}). As with fixed-point computation, known algorithms for Nash equilibrium computation exhibit exponential dependence on the dimension. In two-player general-sum games, exact equilibria can be computed in time $2^{O(n)}$, where $n$ is the number of strategies per player. Beyond two-player games, the worst-case state-of-the-art has an even steeper dependence on the problem parameters. Specifically, for computing $\eps$-approximate Nash equilibria in games with a constant number of players and $n$ strategies per player, existing algorithms achieve the following runtime tradeoffs: $\log(1/\eps) \cdot n^{O(n)}$ using algorithms from the existential theory of the reals to solve the polynomial equations and inequalities defining Nash equilibria~\cite{renegar1992computational}; $(1/\eps)^{O(n)}$ using cover-based methods; $2^{O(\log^2 n/\eps^2)}$ by enumerating small-support strategies, following Lipton, Markakis, and Mehta~\cite{lipton2003playing}; and $2^{O(n)}\mathrm{poly}(\log(1/\eps))$ using the algebraic algorithm of Lipton and Markakis~\cite{lipton2004nash}, suppressing polynomial dependence on the payoff bit length.

We are interested in whether smoothed analysis models can be leveraged to improve upon these running times. Although smoothed analysis has been studied in the context of Nash equilibrium computation, it is known that finding a Nash equilibrium remains PPAD-hard under random perturbations of magnitude $1/\mathrm{poly}(n)$ to the payoff entries of two-player games~\cite{chen2006computing}, and even under constant-magnitude perturbations (under randomized reductions)~\cite{boodaghians2020smoothed}. This hardness result, however, does not preclude smoothed analysis from substantially improving the best-known exponential upper bounds.
\begin{question}
Can we exploit a smoothed analysis perspective to obtain algorithms that improve upon the worst-case state-of-the-art for approximate Nash equilibrium computation?
\end{question}
\paragraph{This Work.} In this work, we study three closely related problems under smoothed analysis: (i)~computing approximate Nash equilibria of normal-form games with a constant number of players and $n$ strategies per player; (ii)~computing approximate fixed points of continuous maps $F$ from the $n$-dimensional Euclidean unit ball $\mathcal{B}_n(0,1)$ to itself; and (iii)~computing approximate solutions to variational inequality (VI) problems defined by continuous maps $F$ from $\mathcal{B}_n(0,1)$ to $\mathbb{R}^n$~\cite{kinderlehrer2000introduction}. Variational inequalities, which are formally defined in Section~\ref{subsec:problem-overview}, constitute a fundamental variant of the fixed-point problem.

\paragraph{Our Smoothed Analysis Model.} We use the same type of smoothed-analysis model for all three problems, which applies to any perturbation distribution under which the random quantities evaluated by our algorithms---such as payoffs in games, and function values (and possibly also gradients) in fixed-point problems---have bounded density. This captures, for example, worst-case games perturbed by small random per-action costs, and worst-case maps perturbed by a random affine map (see Section~\ref{subsec:problem-overview}).

Our main contributions are as follows:
\begin{itemize}[leftmargin=18pt]

    \item \textbf{Smoothed Analysis for Fixed Points and Variational Inequalities via Recursive Cover-Based Search (Theorem~\ref{thm:smoothed}).} For bounded and Lipschitz maps $F$ on $\mathcal{B}_n(0,1)$ whose values have bounded density (Assumption~\ref{asm:bounded-density}), we design a recursive brute-force (tree-search and pruning) algorithm using only zeroth-order oracle access to $F$ that computes an $\eps$-approximate fixed point or VI solution in expected time $2^{O(n)} \cdot \log(1/\eps)$ (for constant Lipschitz and density parameters), improving exponentially in $1/\eps$ over the worst-case bound of $(1/\eps)^{O(n)}$. The algorithm builds on an idea suggested by an anonymous reviewer of a previous version of this paper.

    \item \textbf{Smoothed Analysis for Fixed Points and Variational Inequalities via Homotopy Path-Following (Theorem~\ref{ithm:smoothed}).} When $F$ is additionally smooth (i.e., has a Lipschitz Jacobian) and its directional derivatives also have bounded density (Assumption~\ref{asm:bounded-jacobian-density}), we analyze the standard predictor--corrector homotopy path-following method~\cite{allgower2003introduction}, applied to the classical linear homotopy between the identity and $F$~\cite{kellogg1976constructive,chow1978finding}, using zeroth- and first-order oracle access to $F$. With probability at least $1-p$, our homotopy algorithm finds an $\eps$-approximate fixed point or VI solution in time $2^{O(n)} \cdot \mathrm{poly}(1/p) \cdot (1/\eps)$ (for constant regularity and density parameters). Moreover, by injecting a small artificial Gaussian affine perturbation, this homotopy method also yields a randomized worst-case algorithm with runtime $(1/\eps)^{O(n)}$ (Theorem~\ref{thm:algorithm final}). 

    \item \textbf{Smoothed Analysis for Nash Equilibria (Theorem~\ref{thm:overview-main}).} We provide a deterministic recursive grid-refinement algorithm that uses $2^{O(n)} \cdot \log(1/\eps)$ expected payoff-vector computations (and expected running time $T \cdot 2^{O(n)} \cdot \log(1/\eps)$) to find an $\eps$-approximate Nash equilibrium in a game with a constant number of players and $n$ strategies per player, when the density parameter in Assumption~\ref{asm:bounded-density games} satisfies $M=O(1)$. For comparison, known worst-case algorithms achieve $\log(1/\eps)\cdot n^{O(n)}$~\cite{renegar1992computational}, $(1/\eps)^{O(n)}$ (brute force),  $2^{O(\log^2 n/\eps^2)}$~\cite{lipton2003playing},  and $2^{O(n)}\mathrm{poly}(\log(1/\eps))$ \cite{lipton2004nash}. This builds on the cover-based algorithm for fixed points.

    \item \textbf{Lower Bound (Theorem~\ref{ithm:lower}).} To show that the exponential dependence on $n$ in Theorem~\ref{thm:smoothed} is necessary, we note that there is a $2^{\Omega(n)}$ query lower bound for finding $O(1)$-approximate fixed points of $O(1)$-Lipschitz endomorphisms of $\mathcal{B}_n(0,1)$, and hence for $O(1)$-approximate VI solutions, that holds even under the smoothed-analysis model of Theorem~\ref{thm:smoothed}. This lower bound follows directly from Rubinstein's construction~\cite{rubinstein2017settling} via a simple scaling argument, and we include the short derivation for completeness.
\end{itemize}

\paragraph{Relationship Between Our Two Fixed-Point Algorithms.}
Recursive cover-based search (Theorem~\ref{thm:smoothed}) achieves a better dependence on $\varepsilon$ under weaker assumptions than our homotopy guarantee (Theorem~\ref{ithm:smoothed}). As highlighted in Question~\ref{q:paradigms}, our primary goal is not merely to optimize asymptotic bounds, but to examine whether classical algorithmic paradigms provably outperform their worst-case behavior under smoothed analysis. Whereas cover-based search refines a frontier of candidate points, homotopy continuation tracks a single solution curve using local first-order information~\cite{kellogg1976constructive,chow1978finding,watson1987hompack}. In practice, homotopy methods are the standard choice for smooth problems in moderate dimension~\cite{watson1987hompack,allgower2003introduction}, whereas cover-based algorithms are inherently exponential-time. To the best of our knowledge, our analysis gives the first explicit runtime guarantee for approximate fixed-point computation of general smooth maps via a homotopy-type method, made possible by our smoothness assumption.

\subsection{Main Results}\label{subsec:problem-overview}

For convenience of presentation, we split our results below into subsections covering fixed points and variational inequalities, followed by Nash equilibria in normal-form games. In each subsection, we first introduce the problem setting and perturbation model, and then state our main results.

\subsubsection{Fixed points and Variational Inequalities}\label{sec:fp_vi_intro}

We first study the problems of finding approximate fixed points and solving variational inequalities over the $n$-dimensional Euclidean unit ball. Unless otherwise stated, all norms $\|\cdot\|$ in this paper refer to the Euclidean ($\ell_2$) norm, and we write $\mathcal{B}_n(0,1) = \mathcal{B}(0,1) = \mathcal{B}_n := \{x \in \mathbb{R}^n \colon \|x\|_2 \le 1\}$.
\begin{definition}[Fixed Point (FP) Problem]\label{def:fixed point problem} 
Let ${F} \colon \mathcal{B}_n(0,1) \to \mathcal{B}_n(0,1)$ be a continuous function. We say that $x \in \mathcal{B}_n(0,1)$ is an $\varepsilon$-approximate fixed point if
\[
\| F(x)-x\|_2 \le \varepsilon~.
\]
\end{definition}

\begin{definition}[Variational Inequality (VI) Problem]\label{def:target} 
Let ${F} \colon \mathcal{B}_n(0,1) \to \mathbb{R}^n$ be a continuous function. We say that $x \in \mathcal{B}_n(0,1)$ is an $\varepsilon$-approximate solution to the variational inequality for $F$ if for all $y \in \mathcal{B}_n(0,1)$,
\[
\langle F(x), y-x\rangle \le \varepsilon~.
\]
\end{definition}
Variational inequalities give a common formulation for constrained optimization and equilibrium problems~\cite{kinderlehrer2000introduction,facchinei2003finite}. For example, if $F=-\nabla f$ for a smooth function $f$, then an $\varepsilon$-approximate VI solution is an approximate first-order stationary point of $f$ over $\mathcal{B}_n(0,1)$, and an approximate minimizer when $f$ is convex. Similarly, if $F$ stacks the gradients of each player's utility with respect to their own strategy, then VI solutions are the (first-order) Nash equilibria of the corresponding continuous game~\cite{facchinei2003finite}. Intuitively, an exact solution $x$ either lies in the interior of the ball with $F(x)=0$, or lies on its boundary with $F(x)$ pointing outward along $x$. Variational inequalities are closely related to fixed points: $x$ is a fixed point of $G$ if and only if it solves the VI for $F(x)=G(x)-x$, and conversely, $x$ solves the VI for $F$ if and only if it is a fixed point of $G(x)=\Pi_{\mathcal{B}_n(0,1)}(x+F(x))$, where $\Pi$ denotes Euclidean projection. We use these two-way reductions in Corollary~\ref{cor:reverse-vi-fp} and Section~\ref{sec:path-following}.

\paragraph{Smoothed Analysis Model for FPs and VIs.} We consider the following smoothed analysis assumptions for studying Problems~\ref{def:fixed point problem} and~\ref{def:target}.
\begin{assumption}[Bounded Density]\label{asm:bounded-density}
    There exists some $P_v > 0$ such that, for all $x \in \mathcal{B}_n(0,1)$, the distribution of $F(x)$ has probability density upper bounded by $P_v$. 
\end{assumption}

\noindent When analyzing homotopy-based path-following methods for these problems, we also impose a joint density bound on the function value and its directional derivatives.
\begin{assumption}[Bounded Joint Density for Derivative]\label{asm:bounded-jacobian-density}
    There exists some $P_d > 0$ such that, for all $x \in \mathcal{B}_n(0,1)$ and all unit vectors $v \in \mathbb{S}^{n-1}$, the joint distributions of $(F(x),v^\top J_F(x))$ and $(F(x), J_F(x)v)$ both have probability densities upper bounded by $P_d$, where $J_F(x)$ denotes the Jacobian of $F$ at $x$.
\end{assumption}

\noindent We also use the following regularity assumptions on $F$, together with the oracle model below.
\begin{assumption}[Bounded norm]\label{asm:F-bounded}
    $F$ is bounded: $\|F(x)\|\le L_0$ for all $x\in\cB(0,1)$.
\end{assumption}
\begin{assumption}[Lipschitzness]\label{asm:F-Lip}
$F$ is $L_1$-Lipschitz: $\forall x,y\in\cB(0,1),\ \|F(x)-F(y)\|\le L_1\|x-y\|$. If $F$ is continuously differentiable, this is equivalent to $\|J_F(x)\|_{\mathrm{op}}\le L_1$ for all $x\in\cB(0,1)$, where $J_F(x)$ is the Jacobian of $F$ at $x$.
\end{assumption}
\begin{assumption}[Smoothness: Lipschitzness of the Jacobian]\label{asm:F-smooth}
    The Jacobian $J_{F}(x)$ of $F$ is $L_2$-Lipschitz: $\forall x,y\in\cB(0,1),\ \|J_{F}(x)-J_{F}(y)\|_{\mathrm{op}}\le L_2 \|x-y\|$. 
\end{assumption} 
\begin{assumption}[Zeroth-order and first-order oracle]\label{asm:F-zero-first-order}
    A zeroth-order oracle provides access to the values $F(x)$. A first-order oracle provides access to the Jacobian $J_F(x)$. 
\end{assumption}

We now describe a concrete family of perturbation models that satisfies our assumptions.

\begin{example}[Random affine perturbations for VIs]\label{ex:affine-vi}
Let $F_0:\mathcal{B}_n(0,1) \rightarrow \mathbb{R}^n$ satisfy Assumptions~\ref{asm:F-bounded}, \ref{asm:F-Lip}, and~\ref{asm:F-smooth} for parameters $L_0, L_1, L_2$, and let
\[
F(x) = F_0(x) + D x + e,
\]
where $D=(\Delta_{ij}/\sqrt{n})$ is an $n\times n$ matrix and $e=(E_i/\sqrt{n})$ is a vector, with $\Delta_{ij}$ i.i.d.\ and $E_i$ i.i.d. The random variables $\Delta_{ij}$ and $E_i$ may follow centered, density-bounded distributions supported on a compact set, or more generally centered, density-bounded sub-Gaussian distributions (such as Gaussians). Here $F_0$ may be fixed or random, and $F_0, D, e$ are mutually independent. Then $F$ satisfies Assumptions~\ref{asm:bounded-density} and~\ref{asm:bounded-jacobian-density}, and Assumptions~\ref{asm:F-bounded}, \ref{asm:F-Lip}, and~\ref{asm:F-smooth} hold with high probability, or deterministically in the compactly supported case (Corollary~\ref{cor:instantiate}).
\end{example}

For fixed points, this construction does not directly apply, since $F$ must also map $\mathcal{B}_n(0,1)$ into itself. Instead, we can take $F_0$ with $\|F_0(x)\|\le r_0 < 1$ for all $x\in\cB(0,1)$ (e.g., $F_0=r_0G_0$ for any $G_0:\cB(0,1)\to\cB(0,1)$ satisfying Assumptions~\ref{asm:F-Lip} and~\ref{asm:F-smooth}), and sample a bounded random matrix $D$ and vector $e$ satisfying $\|D\|_{\mathrm{op}}+\|e\|\le 1-r_0$; Assumptions~\ref{asm:bounded-density} and~\ref{asm:bounded-jacobian-density} then hold for suitable distributions of $D$ and $e$ (Corollary~\ref{cor:instantiate2}).

We establish two upper bounds for fixed points and variational inequalities: Theorem~\ref{thm:smoothed} uses a recursive cover-based algorithm, and Theorem~\ref{ithm:smoothed} analyzes the standard predictor--corrector homotopy method under slightly stronger assumptions (smoothness, bounded Jacobian density, and first-order access). Cover-based search inherently explores exponentially many regions, whereas homotopy methods follow a single solution path and underlie standard numerical software~\cite{watson1987hompack}.

\begin{theorem}[Fixed Points and VIs, cover-based; informal version of Theorem~\ref{thm:brute-force}] \label{thm:smoothed}
    Under Assumptions~\ref{asm:F-bounded}, \ref{asm:F-Lip}, and~\ref{asm:bounded-density} (with $L_0=1$ for fixed points), there is a deterministic algorithm with zeroth-order access to $F$ that, for any $0<\varepsilon\le 1$, always returns an $\varepsilon$-approximate solution to the fixed-point or VI problem, in expected time
    \[
    \cC^{n}\cdot P_v\cdot \mathrm{Vol}(\mathcal B_n) \cdot\left(1+\log\left(\frac{1}{\varepsilon}\right)\right), \qquad \text{where } \cC = O\bigl((1+L_0)(1+L_1)\bigr).
    \]
\end{theorem}

\begin{theorem}[Fixed Points and VIs, homotopy-based; informal version of Theorem~\ref{thm:smooth analysis final}] \label{ithm:smoothed}
    Under the assumptions of Theorem~\ref{thm:smoothed} together with Assumptions~\ref{asm:F-smooth} and~\ref{asm:bounded-jacobian-density}, there is a deterministic algorithm with zeroth- and first-order access to $F$ that, for all $0<p<1$, $0<\varepsilon<1$, and $n\ge 2$, returns an $\varepsilon$-approximate solution with probability at least $1-p$ in time
    \[
\cC_1^n{\left(\frac{P_d}{{n}^n}\cdot \frac{P_v}{\sqrt{n}^n}\cdot \frac{1}{p}\right)^{\cC_2}}\cdot\frac{1}{\varepsilon}, \qquad \text{where } \cC_1 = (1+L_0+L_1+L_2)^{O(1)},\ \cC_2 = O(1).
\]
\end{theorem}

\begin{remark}[Source of randomness]
Both algorithms are deterministic. The expectation in Theorem~\ref{thm:smoothed} and the probability in Theorem~\ref{ithm:smoothed} are over the random draw of $F$ from the smoothed distribution.
\end{remark}

\begin{remark}[Worst-case instances]
Theorem~\ref{ithm:smoothed} also yields a randomized algorithm for worst-case instances satisfying only Assumptions~\ref{asm:F-bounded}, \ref{asm:F-Lip}, and~\ref{asm:F-smooth}: adding a small Gaussian affine perturbation $Dx+e$ of magnitude $O(\varepsilon)$ to $F$ preserves $\varepsilon$-approximate solutions with high probability, and the perturbed instance satisfies our density assumptions, giving running time $(\cC/\varepsilon)^{O(n)}$ (Theorem~\ref{thm:algorithm final}).
\end{remark}

Finally, the dimension dependence in Theorem~\ref{thm:smoothed} is near-optimal. This follows from Rubinstein's construction~\cite{rubinstein2017settling}, combined with a simple scaling and perturbation argument that we include in Section~\ref{sec:lowerBound} for completeness. This yields the lower bound:

\begin{theorem}[Lower bound, Fixed Points/VI, from~\cite{rubinstein2017settling}. Informal version of Theorem~\ref{thm:lowerBoundMain}] \label{ithm:lower}
Let $ G : \mathcal{B}_n(0,1) \to \mathcal{B}_n(0,1) $ be an $ O(1) $-Lipschitz function. Then, in the worst case, $ 2^{\Omega(n)} $ value queries are required to find an $ O(1) $-approximate fixed point of $ G $. This implies a $2^{\Omega(n)}$ query lower bound for finding $O(1)$-approximate fixed points or variational inequality solutions over $\mathcal{B}_n(0,1)$ in the smoothed-analysis setting of Theorem~\ref{thm:smoothed}.
\end{theorem}

\subsubsection{Nash Equilibria in Normal-Form Games}
We next study finite normal-form games with $m$ players, where each player has $n$ possible actions. Such a game is specified by payoff functions
$
u_i:[n]^m \to [0,1],\ i\in[m],
$
where $u_i(a_1,\dots,a_m)$ is the payoff to player $i$ when the players choose the action profile $(a_1,\dots,a_m)$.

A mixed strategy for player $i$ is a probability distribution over their $n$ actions. A mixed strategy profile is a tuple
$
x=(x_1,\dots,x_m)\in (\Delta_n)^m,
$
where $\Delta_n$ denotes the probability simplex over $[n]$. We write $u_i(x)$ for the expected payoff of player $i$ when each player independently samples an action from their own mixed strategy $x_i$.

\paragraph{Approximate Nash equilibrium.}
A mixed strategy profile $x$ is an $\eps$-Nash equilibrium if no player can gain more than $\eps$ by unilaterally deviating to another strategy while the other players keep theirs fixed. Equivalently, for every player $i \in [m]$ and every pure action $a\in[n]$,
\[
u_i(a,x_{-i}) \le u_i(x)+\eps,
\]
where $u_i(a,x_{-i})$ denotes the expected utility to player $i$ when playing pure action $a$ while the remaining players play according to $x_{-i}$. When $\eps=0$, $x$ is an exact Nash equilibrium.

\paragraph{Smoothed Analysis Model for Games.} In our smoothed analysis model for games, we require that for every mixed strategy profile $x$, the vector of all players' single-action deviation payoffs has bounded probability density. Formally, our model is defined as follows.
\begin{assumption}[Bounded Density Games]\label{asm:bounded-density games} \label{ass:density}
    There exists some $M \ge 1$ such that, for every mixed strategy profile $x \in (\Delta_n)^m$, the single-action deviation payoff vector
    \[
\mathcal U(x):=\bigl(U_{i,a}(x)\bigr)_{i\in[m],\,a\in[n]} \in [0,1]^{mn}, \qquad \text{where } U_{i,a}(x) := u_i(a,x_{-i}),
\]
    is a random vector whose probability density on $[0,1]^{mn}$ is everywhere at most $M^{mn}$.
\end{assumption}
One natural way to satisfy this assumption is to consider a random game, or to start from an arbitrary baseline game and add a random cost to each action of each player---i.e., shift the slices of each player's payoff tensor indexed by their own actions by small random perturbations, subject to the density and payoff-range conditions in Corollary~\ref{cor:costs}. Informally, this condition ensures that after perturbation, the deviation-payoff vectors cannot concentrate excessively in any tiny region. Our main theorem for Nash equilibrium computation under this smoothed analysis model is the following.

\begin{theorem}[Nash Equilibria; informal version of \Cref{thm:main-nash}]\label{thm:overview-main}
Assume \Cref{ass:density}, and suppose that for any mixed profile $x$ one can compute the full deviation-payoff vector
$
\mathcal U(x)=\bigl(U_{i,a}(x)\bigr)_{i\in[m],\,a\in[n]}
$
in time $T$. Then there exists a universal constant $\cC>1$ such that for every $\eps\in(0,1]$, there is a deterministic algorithm (using no internal randomness) that always returns an $\eps$-approximate Nash equilibrium of the realized game. Its expected number of payoff-vector computations, where the expectation is over the random game, is at most
\[
\Bigl(1+\Bigl\lceil \log_2\frac{4m}{\eps}\Bigr\rceil\Bigr)(\cC mM)^{mn}.
\]
and the expected runtime is $
O\!\left(T(\cC mM)^{mn}\log\frac{m}{\eps}\right).
$.
\end{theorem}

The next corollary instantiates this theorem for games with independent random payoffs (see also Corollary~\ref{cor:costs} in Section~\ref{sec:nash} for the random action-cost perturbation model).

\begin{corollary}[Nash with Random Payoffs]\label{cor:overview-iid}
Assume that, after normalization to $[0,1]$, all payoff entries
\[
\bigl\{u_i(a_1,\dots,a_m):i\in[m],\ (a_1,\dots,a_m)\in[n]^m\bigr\}
\]
are independent, each with a probability density bounded by $M'$. Then \Cref{ass:density} holds with
$
M=\sqrt{2}\,M' n^{(m-1)/2}.
$
Consequently, there is a deterministic algorithm whose expected number of payoff-vector computations is at most
$
O\!\left(\left(\cC M'mn^{(m-1)/2}\right)^{mn}\log \frac{4m}{\eps} \right),
$
and whose expected running time is at most
$
O\!\left(T\left(\cC M'mn^{(m-1)/2}\right)^{mn}\log \frac{4m}{\eps} \right),
$
where $\cC$ is a universal constant.
\end{corollary}

\subsection{Proof Overview for Theorem~\ref{ithm:smoothed}}

We describe the algorithm for variational inequalities; for a fixed-point problem with map $G$, we apply it to the variational inequality with $F(x)=G(x)-x$ (Corollary~\ref{cor:instantiate2}). The predictor--corrector homotopy method analyzed in Theorem~\ref{ithm:smoothed} follows a path of zeros of the normalized homotopy function $K : \mathcal{B}_n(0,1) \times [0,1] \to \mathbb{R}^n$,
\[
K(x,t)=(1-t)x-t F(x)/a,
\]
where we set the normalization factor $a=L_0+L_1$ so that $K$ has an $O(1)$ Lipschitz constant. For a point $(x,t)$, we refer to $x$ as its spatial component. Since $K(x,0)=x$, $(0,0)$ is the unique zero of $K$ on the slice $t=0$. Let $\gamma$ denote the branch of the zero set $\{(x,t) : K(x,t)=0\}$ that leaves $(0,0)$ into $t>0$, followed until it first reaches either $t=1$ or $\|x\|=1$. Our algorithm tracks $\gamma$ from $(0,0)$ to this endpoint and outputs its spatial component.

\paragraph{Why the path ends at a solution.}
A point $x \in \mathcal{B}_n(0,1)$ satisfies the variational inequality exactly (Definition~\ref{def:target} with $\varepsilon=0$) if and only if at least one of the following holds: either $F(x) = 0$; or $\|x\|=1$ and $F(x) = \lambda x$ for some $\lambda \ge 0$.

At $t=1$, $K(x,1)=0$ means $F(x)=0$. For $t\in(0,1)$, $K(x,t)=0$ is equivalent to $F(x)=\lambda x$ with $\lambda=a(1-t)/t>0$. Hence, if $\gamma$ reaches $t=1$, its spatial component is a zero of $F$, and if it reaches $\|x\|=1$, its spatial component is a boundary solution. Figure~\ref{fig:homotopy} illustrates both cases for $n=1$.

\begin{figure}[t]
\centering
\begin{tikzpicture}[scale=0.72,>=stealth]
  \fill[gray!8] (0,-2) rectangle (5,2);
  \draw[dashed] (0,2) -- (5,2) -- (5,-2);
  \draw[->] (0,-2) -- (0,2.6) node[above] {$x$};
  \draw[->] (0,-2) -- (5.7,-2) node[right] {$t$};
  \foreach \y/\l in {-2/-1,0/0,2/1} \node[left,font=\small] at (0,\y) {$\l$};
  \foreach \x/\l in {0/0,5/1} \node[below,font=\small] at (\x,-2) {$\l$};
  \draw[very thick,blue] plot[smooth,tension=0.7] coordinates
    {(0,0) (1,0.45) (2.2,0.7) (3.3,0.95) (3.75,1.3) (3.4,1.65) (3.15,2)};
  \fill (0,0) circle (2.2pt) node[below right,font=\small] {start};
  \fill[red] (3.15,2) circle (2.2pt);
  \node[blue] at (2.2,0.35) {$\gamma$};
  \node[red,align=center,font=\small] at (1.5,1.5) {end: $|x|=1$};
  \draw[red,->,thin] (2.3,1.6) -- (3.0,1.95);
  \node at (2.5,-3.0) {(a) $\gamma$ ends with $|x|=1$};
\end{tikzpicture}
\hspace{0.8cm}
\begin{tikzpicture}[scale=0.72,>=stealth]
  \fill[gray!8] (0,-2) rectangle (5,2);
  \draw[dashed] (0,2) -- (5,2) -- (5,-2);
  \draw[->] (0,-2) -- (0,2.6) node[above] {$x$};
  \draw[->] (0,-2) -- (5.7,-2) node[right] {$t$};
  \foreach \y/\l in {-2/-1,0/0,2/1} \node[left,font=\small] at (0,\y) {$\l$};
  \foreach \x/\l in {0/0,5/1} \node[below,font=\small] at (\x,-2) {$\l$};
  \draw[thick,gray] plot[smooth,tension=0.8] coordinates
    {(5,-0.45) (4.3,-0.85) (5,-1.45)};
  \node[gray,font=\small,align=center] at (3.1,-1.4) {other zeros\\ of $K$};
  \draw[very thick,blue] plot[smooth,tension=0.7] coordinates
    {(0,0) (1,0.9) (2.2,0.45) (3.3,0.15) (4.2,0.45) (5,0.6)};
  \fill (0,0) circle (2.2pt) node[below right,font=\small] {start};
  \fill[red] (5,0.6) circle (2.2pt);
  \node[blue] at (2.2,-0.05) {$\gamma$};
  \node[red,font=\small] at (3.4,1.4) {end: $t=1$};
  \draw[red,->,thin] (4.1,1.25) -- (4.9,0.72);
  \node at (2.5,-3.0) {(b) $\gamma$ ends with $t=1$};
\end{tikzpicture}
\caption{The homotopy path for $n=1$, drawn in the $(t,x)$ plane. The path $\gamma$ (blue) is the branch of the zero set $\{(x,t):K(x,t)=0\}$ that starts at $(0,0)$; it may turn back in $t$. (a) $\gamma$ reaches $|x|=1$ at some $t<1$, where $F(x)=\lambda x$ with $\lambda\ge0$, giving a boundary solution. (b) $\gamma$ reaches $t=1$, where $F(x)=0$. Other components of the zero set (gray) are not followed.}
\label{fig:homotopy}
\end{figure}

\paragraph{Regularity of the path.}
Both the algorithm and its analysis rely on the homotopy Jacobian $J_K=[\partial_x K \;\; \partial_t K]$, the $n\times(n+1)$ matrix of partial derivatives of $K$ with respect to both $x$ and $t$, being well conditioned along the zero set of $K$. For an $n\times(n+1)$ matrix $A$, we denote by $\sigma_n(A)=\min_{\|v\|=1}\|A^\top v\|$ its smallest singular value, which we call its \emph{smallest row singular value}; it is positive if and only if $A$ has full row rank $n$.

\begin{property}[Row singular value bound of the homotopy Jacobian]\label{prop:spectral}
	At every zero $z=(x,t) \in \mathcal{B}_n(0,1)\times[0,1]$ of $K$, the smallest row singular value of $J_K(z)$ is at least some $\theta>0$.
\end{property}

Under our smoothed analysis model, Property~\ref{prop:spectral} holds with $\theta=\cC^{-n}p$ with probability at least $1-p$ (Lemma~\ref{lem:delta}); we sketch this argument after describing the algorithm and its guarantees. When it holds, the implicit function theorem implies that $\gamma$ is a regular curve whose curvature is upper bounded by $O(1/\theta)$ (for fixed smoothness parameters).

\paragraph{Algorithm idea.}
Because an algorithm cannot take infinitesimal steps along $\gamma$, it does not stay exactly on the curve, but rather tracks $\gamma$ closely via discrete predictor--corrector steps.

We provide an informal summary of the procedure in Algorithm~\ref{alg:discrete-informal}. For each iteration $i = 0, 1, \dots$, let $x_i=(\bar x_i,t_i)\in\mathbb R^{n+1}$ denote the iterate tracking the path $\gamma$, where $\bar{x}_i \in \mathbb{R}^n$ is its spatial component. Since $K=0$ along $\gamma$, differentiating gives $J_K(p)\gamma'(p)=0$ at every $p\in\gamma$, so the tangent lies in $\ker J_K(p)$, which is one-dimensional since $J_K(p)$ has full rank by Property~\ref{prop:spectral}. Since $x_i$ lies close to $\gamma$, we compute a unit vector $v_i \in \ker J_K(x_i)$ (oriented consistently with the previous step $v_{i-1}$, i.e., $\langle v_i, v_{i-1}\rangle > 0$) to approximate the path tangent. The iteration starts at $x_0=(0,0)$, where the tangent is known explicitly: $J_K(0,0)=[I \;\; -F(0)/a]$, whose kernel is spanned by $(F(0)/a, 1)$, and we take $v_{-1}$ to be this vector normalized, which points into $t>0$. Given a step-size parameter $h > 0$, each iteration consists of two steps (Figure~\ref{fig:predictor-corrector}):
\begin{itemize}
    \item \emph{Predictor step:} move along the approximate tangent, $y_i=x_i+hv_i$. Since $\gamma$ curves, $y_i$ is slightly off the zero set of $K$.
    \item \emph{Corrector step:} apply one Newton step to return toward the zero set, $x_{i+1} = y_i + c_i$, where $c_i \perp v_i$ solves $J_K(y_i)c_i = -K(y_i)$.
\end{itemize}
\begin{figure}[t]
\centering
\begin{minipage}[c]{0.44\textwidth}
\centering
\begin{tikzpicture}[scale=0.72,>=stealth]
  \draw[very thick,blue,domain=1.3:8.7,smooth,samples=60] plot (\x,{2.4-0.15*(\x-5)^2});
  \node[blue] at (8.9,0.6) {$\gamma$};
  \coordinate (x0) at (1.6,0.666);
  \coordinate (y0) at (3.42,2.523);
  \coordinate (x1) at (3.776,2.174);
  \coordinate (y1) at (6.217,3.070);
  \coordinate (x2) at (6.605,2.014);
  \draw[->,thick] (x0) -- (y0) node[midway,above left,font=\small] {predictor};
  \draw[->,thick,dashed,red] (y0) -- (x1);
  \draw[->,thick] (x1) -- (y1);
  \draw[->,thick,dashed,red] (y1) -- (x2) node[midway,right,font=\small] {corrector};
  \foreach \p in {x0,x1,x2} \fill (\p) circle (1.8pt);
  \foreach \p in {y0,y1} \filldraw[fill=white,draw=black] (\p) circle (1.8pt);
  \node[below right,font=\small] at (x0) {$x_i$};
  \node[above,font=\small] at (y0) {$y_i$};
  \node[below,font=\small] at (x1) {$x_{i+1}$};
  \node[above,font=\small] at (y1) {$y_{i+1}$};
  \node[below,font=\small] at (x2) {$x_{i+2}$};
\end{tikzpicture}
\end{minipage}\hfill
\begin{minipage}[c]{0.52\textwidth}
\caption{Two iterations of Algorithm~\ref{alg:discrete-informal}: each predictor step (solid) moves along the approximate tangent and leaves $\gamma$; each corrector step (dashed) is a Newton step orthogonal to it that returns close to $\gamma$.}
\label{fig:predictor-corrector}
\end{minipage}
\end{figure}
\textbf{Stopping rule:} the predicate in Lemma~\ref{asm: oracle} checks at each step whether $\bar x_i$ is already an $\varepsilon$-approximate interior zero ($\|F(\bar{x}_i)\| \le \varepsilon/2$) or can be radially projected onto the unit sphere ($\bar{x}_i / \|\bar{x}_i\|$) to yield an $\varepsilon$-approximate boundary VI solution.
\begin{algorithm}[h]
\caption{Follow-The-Path (Informal version of Algorithm~\ref{alg:discrete})}
\label{alg:discrete-informal}
\begin{algorithmic}[1]
\Require Function $F$
\If{$\|F(0)\|\le\varepsilon/2$}
\State \Return $0$
\EndIf
\State $a\gets L_0+L_1$, $K(x,t)\gets(1-t)x-tF(x)/a$
\State $x_0\gets(0,0)$,
$v_{-1}\gets(F(0)/a,1)/\sqrt{1+\|F(0)\|^2/a^2}$
\State $i\gets0$
\While{the stopping condition is not met}
\State Choose a unit vector $v_i\in\ker J_K(x_i)$ with $\langle v_i,v_{i-1}\rangle>0$; $y_i\gets x_i+hv_i$ \Comment{Predictor step}
\State Solve $J_K(y_i)c_i=-K(y_i)$ with $v_i^\top c_i=0$; $x_{i+1}\gets y_i+c_i$ \Comment{Corrector step}
\State $i\gets i+1$
\EndWhile
\State \Return $\bar x_i$ or $\bar x_i/\|\bar x_i\|$ depending on the ending condition
\end{algorithmic}
\end{algorithm}
\paragraph{Why the corrector step works.}
The corrector is a Newton step: $c_i$ solves the linearization $K(y_i)+J_K(y_i)c=0$ of $K(y_i+c)=0$ in the directions orthogonal to the path. When $J_K$ is well conditioned (Property~\ref{prop:spectral}) and $h$ is small relative to $\theta$, this shrinks the residual $\|K\|$ quadratically, bringing the iterate back close to $\gamma$.

Property~\ref{prop:spectral} ensures that each predictor step does not deviate too far from $\gamma$ and that the single Newton correction step contracts the distance back to the curve while maintaining steady forward progress along $\gamma$. Using the $O(1/\theta)$ curvature bound, we prove in Section~\ref{sec:discrete} that choosing step size $h=\Theta(\min(\varepsilon,\theta))$ guarantees that each predictor--corrector cycle advances by $\Theta(h)$ in arc length along $\gamma$ while keeping $x_i$ within $O(h^4/\theta^3)$ of the curve. Thus, if the total arc length of $\gamma$ is $T$, Algorithm~\ref{alg:discrete-informal} terminates with a valid $\varepsilon$-approximate VI solution in $O(1+\max(T/\eps, T/\theta))$ iterations.

\paragraph{Proving Property~\ref{prop:spectral}.}
Under our smoothed analysis model, we establish Property~\ref{prop:spectral} with high probability via a net and union-bound argument over the cylinder $\mathcal{B}_n(0,1) \times [0,1]$. For $t\le 1/4$, regularity holds deterministically because $\partial_x K(x,t) = (1-t)I - tJ_F(x)/a$ has minimum singular value at least $1 - t(1 + L_1/a) \ge 1/2$. For $t\ge 1/4$, where the random term $tF(x)/a$ is large enough for the density bounds to apply, we construct a $\delta$-net $\mathcal{N}$ of $\cB(0,1)\times[1/4,1]$, i.e., a finite collection of points such that every point in $\cB(0,1)\times[1/4,1]$ is within distance $O(\delta)$ of some center in $\mathcal{N}$, and proceed in three steps:

\begin{enumerate}[label={(\arabic*)}]
    \item First, we construct the $\delta$-net $\mathcal{N}$ of size at most $\cC^{n}/\delta^{n+1}$ (Lemma~\ref{lem:net}).

    \item Second, for each fixed net point $(x,t)\in\mathcal N$, we bound the joint probability that $\lrnorm{K(x,t)}$ and the smallest row singular value $\sigma_n(J_K(x,t))$ are simultaneously small. Specifically, Lemma~\ref{lem:joint H} shows that for $0<\alpha\le\eta$, the probability that $\lrnorm{K(x,t)}\le\alpha$ and $\sigma_n(J_K(x,t))\le\eta$ is at most $\cC^n\alpha^n\eta^2/\max\{\|x\|,\eta\}$ (for normalized density bounds).
    \begin{itemize}
        \item \textbf{Why this bound holds.} The exponents follow a dimension count: $K(x,t)=0$ imposes $n$ conditions, and rank deficiency of the $n\times(n+1)$ matrix $J_K$ imposes $2$ more, since some unit vector $v$ (with $n-1$ degrees of freedom) must satisfy the $n+1$ equations $v^\top J_K=0$. The resulting codimension $n+2$ exceeds the dimension $n+1$ of the $(x,t)$-space, so such bad zeros are not expected to exist. The proof makes this quantitative via a union bound over a net of candidate vectors $v$, using the bounded joint density of $(F(x), v^\top J_F(x))$ (Assumption~\ref{asm:bounded-jacobian-density}).
    \end{itemize}
    
    \item Finally, we transfer this estimate from the net centers $\mathcal{N}$ to all zeros of $K$ in $\mathcal{B}(0,1)\times[1/4,1]$ using the Lipschitzness of $K$ and $J_K$.
    \begin{itemize}
        \item \textbf{How the transfer works.} The net in (1) partitions $\mathcal{B}(0,1)\times[1/4,1]$ into cells of radius $O(\delta)$. If $K$ has any zero $z$ with $\sigma_n(J_K(z)) < \theta$, then the closest net center $z' \in \mathcal{N}$ must simultaneously satisfy $\|K(z')\| \le O(\delta)$ and $\sigma_n(J_K(z')) \le \theta + O(\delta)$ by the Lipschitzness of $K$ and $J_K$. Choosing $\alpha, \delta = \Theta(\theta)$ and $\eta = 2\theta$, and taking a union bound over $z' \in \mathcal{N}$ using the joint probability bound in (2) and the weighted net sum in Lemma~\ref{lem:net}, one factor of $\theta$ remains after canceling $\delta^{-(n+1)}$ against $\alpha^n \eta^2$.
    \end{itemize}
\end{enumerate}

\noindent This yields the following high-probability guarantee (formalized in Lemma~\ref{lem:delta}):

\begin{informal}
Taking $\theta=\cC^{-n}p$, with probability at least $1-p$, every zero $z=(x,t) \in \mathcal{B}_n(0,1)\times[0,1]$ of $K$ satisfies $\sigma_n(J_K(z)) \ge \theta$.
\end{informal}

\paragraph{Bounding the length of the path.}
The path $\gamma$ consists of two parts: the initial segment with $t\le 1/4$ and the remaining segment in $t\ge 1/4$. For $t\le 1/4$, the contraction property implies that $K(x,t)=0$ defines a unique curve $(x(t),t)$ of deterministically bounded arc length ($< 1$). For $t\ge 1/4$, we bound the arc length using the curvature bound $O(1/\theta)$ implied by Property~\ref{prop:spectral} together with the $\delta$-net $\mathcal{N}$ of scale $\delta = \Theta(\theta)$. By the Lipschitzness of $K$, the cell around a net point $(x,t) \in \mathcal{N}$ can contain a zero of $K$ only if $\|K(x,t)\|\le O(\delta)$. This reduces bounding the path length to a cell-counting problem via the following three steps (Lemma~\ref{lem:length}):
\begin{enumerate}[label={(\arabic*)}]
    \item For a point $p=\gamma(s)$, the \emph{normal disk} of radius $\delta$ at $p$ is the $n$-dimensional open ball $\{p+w: w\perp\gamma'(s),\ \|w\|<\delta\}$, centered at $p$ and orthogonal to the path. By Lemma~\ref{lem:unique proj}, for $\delta \ll \theta$ the normal disks at distinct points of $\gamma$ are pairwise disjoint.
    Consider the tube formed by the normal disks based at points of $\gamma$ in one cell. 
    By Weyl's tube formula~\cite{gray2003tubes}, its volume equals the arc length of $\gamma$ in the cell times $\mathrm{Vol}(\mathcal B_n)\delta^n$. This tube lies in a ball of radius $O(\delta)$ around the cell center, so 
    we can bound the arc length in the cell by $\delta\cdot\cC^n$.
    
    \item For any fixed net center $(x,t) \in \mathcal{N}$ with $t\ge 1/4$, the bounded value density of $F(x)$ implies that $\Pr[\|K(x,t)\|\le O(\delta)] \le \cC^n\delta^{n}$.

    \item Summing over $|\cN|\le \cC^n / \delta^{n+1}$ net points, the expected number of active cells (those with $\|K(x,t)\|\le O(\delta)$) is at most $\cC^n/\delta$.
\end{enumerate}

Multiplying the expected number of active cells ($\cC^n/\delta$) by the maximum arc length per cell ($\delta\cdot \cC^n$), the $\delta$ factors cancel. Thus, conditioned on the regularity event, the expected total length of $\gamma$ is at most $\cC^n$ (after enlarging $\cC$), and Markov's inequality yields a high-probability path-length bound of $\cC^n/p$.

\subsection{Proof Overview for Theorem~\ref{thm:overview-main}}

The algorithm has two layers of search. First, it guesses the support of the equilibrium: for each player $i$, a nonempty set $A_i\subseteq[n]$ of actions that are played with positive probability. There are at most $(2^n-1)^m\le 2^{mn}$ such support guesses.

Fix one support guess $A=(A_1,\dots,A_m)$ and write
\[
K_A:=\Delta(A_1)\times\cdots\times\Delta(A_m),
\qquad
\Delta(A_i):=\{x_i\in\Delta_n: x_i(a)=0 \text{ for all } a\notin A_i\},
\]
equipped with the product $\ell_1$ metric $\|x-y\|_1 = \sum_{i=1}^m \|x_i - y_i\|_1$. For this guess, the algorithm will not search over all mixed strategy profiles in $K_A$ directly. Instead, it maintains, round by round, a sparse frontier $\mathcal F_t(A)\subseteq K_A$ of \emph{candidate centers}, where each center $c\in\mathcal F_t(A)$ represents the $\ell_1$-ball $B_A(c,\rho_t):=\{x\in K_A:\|x-c\|_1\le\rho_t\}$ of radius $\rho_t:=2m\,2^{-t}$. A candidate center is meant to represent the possibility that there exists an exact Nash equilibrium $x^\star=(x_1^\star,\dots,x_m^\star)$ with support $A$ lying inside $B_A(c,\rho_t)$. The algorithm starts with a single coarse center covering $K_A$ and iteratively refines the surviving centers along a nested family of fixed grids, deleting centers whose balls can no longer contain a Nash equilibrium.

Before describing the iterations, we need one geometric primitive on the support polytope $K_A$, which has intrinsic dimension $R=\sum_{i=1}^m(|A_i|-1)$. Let $q_0=\max\{1,\lceil R/m\rceil\}$ and $q_t=2^t q_0$. At each level $t\ge 0$, consider the fixed rational grid
\[
\mathcal N_t(A):=\bigl\{x\in K_A : q_t x_i(a)\in\mathbb Z_{\ge 0} \text{ for all } i\in[m],\ a\in A_i\bigr\}.
\]
Rounding each coordinate of any $x\in K_A$ (except $a=\max A_i$) down to a multiple of $1/q_t$ and adjusting the last coordinate so that the coordinates sum to $1$ shows that $\mathcal N_t(A)$ is a $\rho_t$-net of $K_A$ of size $|\mathcal N_t(A)|\le\cC^{mn}2^{tR}$ (\Cref{lem:subdivision-primitive}). Moreover, these grids form a natural tree: for a point $c\in\mathcal N_t(A)$, its children in $\mathcal N_{t+1}(A)$ are the feasible grid points obtained by replacing each non-maximal numerator $q_t c_i(a)$ with either $2q_t c_i(a)$ or $2q_t c_i(a)+1$. Each grid point has at most $2^R$ children, each child lies in $B_A(c,\rho_t)$, and every point of $\mathcal N_{t+1}(A)$ has a unique parent in $\mathcal N_t(A)$.

Now we describe the iterative search. At round $t=0,1,2,\dots,D$, where $D=\bigl\lceil\log_2\frac{4m}{\eps}\bigr\rceil$, the algorithm stores a collection $\mathcal F_t(A)\subseteq K_A$ of level-$t$ candidate centers. At round $0$, $\mathcal F_0(A)=\{b_A\}$ contains just the barycenter $b_A$ of $K_A$; since $K_A$ has $\ell_1$-diameter at most $2m$, the single ball $B_A(b_A,\rho_0)$ covers all of $K_A$. To go from round $t$ to round $t+1$, the algorithm tests every center $c\in\mathcal F_t(A)$ and keeps only those that remain plausible. If $b_A$ passes at $t=0$, the algorithm initializes $\mathcal F_1(A)=\mathcal N_1(A)$; at every subsequent round $1\le t<D$, each surviving center $c\in\mathcal F_t(A)$ adds all of its children in $\mathcal N_{t+1}(A)$ to $\mathcal F_{t+1}(A)$ (with no duplicates, by uniqueness of parents). After round $D$, the algorithm outputs any surviving center $c\in\mathcal F_D(A)$ that passes the test.

The test comes from a basic property of Nash equilibria. If $x^\star$ is an exact Nash equilibrium and player $i$ puts positive probability on exactly the actions in $A_i$, then:
\begin{enumerate}
    \item all actions in $A_i$ have the same deviation payoff against $x^\star_{-i}$;
    \item every action outside $A_i$ has deviation payoff no larger than this common value.
\end{enumerate}
Equivalently, at equilibrium the player is indifferent among the actions it actually uses, and unused actions are not better.

This suggests the following test. For a level-$t$ candidate center $c\in\mathcal F_t(A)$, we keep $c$ only if for every player $i$, there exists an anchor action $\alpha_i\in A_i$ such that
$U_{i,\alpha_i}(c)\le U_{i,a}(c)\le U_{i,\alpha_i}(c)+2\rho_t$
for all  $a\in A_i$,
and also $
U_{i,b}(c)\le U_{i,\alpha_i}(c)+2\rho_t$
for all $b\notin A_i$.
So, at round $t$, all actions in the guessed support must still look almost tied, and actions outside the guessed support must not look noticeably better. At the final round $D$, the tolerance $2\rho_D=4m\,2^{-D}$ is at most $\eps$, so every surviving center at level $D$ is an $\eps$-approximate Nash equilibrium.

Why does this succeed in finding an $\eps$-approximate Nash equilibrium? Let $x^\star$ be an exact Nash equilibrium of the realized game, and let $A_i$ be the set of actions played with positive probability by $x_i^\star$. For this correct support guess $A=(A_1,\dots,A_m)$, let $c_t\in\mathcal N_t(A)$ ($t\ge 1$) be the grid point obtained by rounding down the non-maximal coordinates of $x^\star$ to multiples of $1/q_t$. By construction, $\|c_t-x^\star\|_1\le\rho_t$, and $c_{t+1}$ is a child of $c_t$ for every $t\ge 1$. Since deviation payoffs are $1$-Lipschitz in the opponents' mixed strategies, each coordinate of the vector $\mathcal U(c_t)=\bigl(U_{i,a}(c_t)\bigr)_{i\in[m],\,a\in[n]}$ differs by at most $\rho_t=2m\,2^{-t}$ from the corresponding coordinate of the exact deviation-payoff vector at $x^\star$. Since the Nash equal-payoff conditions hold exactly at $x^\star$, they hold up to $2\rho_t$ at $c_t$ (and at $b_A$ for $t=0$). Therefore, the chain of centers $b_A,c_1,\dots,c_D$ covering $x^\star$ is never deleted. At depth $D$, $c_D$ is an $\eps$-approximate Nash equilibrium, so the algorithm always returns a valid output.

It remains to bound the running time. Fix a support guess $A$ and a level $t\ge 1$, and write $R=\sum_{i=1}^m(|A_i|-1)$. Every center $c\in\mathcal F_t(A)$ tested at round $t$ lies within distance $\rho_{t-1}=2\rho_t$ of its parent in $\mathcal F_{t-1}(A)$, which passed the round-$(t-1)$ test with tolerance $2\rho_{t-1}=4\rho_t$. By Lipschitzness of the deviation payoffs, any such $c$ must satisfy the test inequalities with relaxed tolerance $8\rho_t$, which imposes $R$ near-equalities of width $8\rho_t$ among the support payoffs. For any fixed grid point $x\in\mathcal N_t(A)$, \Cref{ass:density} implies that this happens with probability at most $M^{mn}(16\rho_t)^R = M^{mn}(32m\,2^{-t})^R$. On the other hand, the total number of level-$t$ grid points is $|\mathcal N_t(A)|\le\cC^{mn}2^{tR}$. Multiplying these two terms by linearity of expectation, the factor $2^{-tR}$ from the survival probability cancels the factor $2^{tR}$ coming from the size of the grid. Hence, the expected number of centers tested at level $t$ for this fixed support guess is at most $(\cC mM)^{mn}$. Finally, summing over the at most $2^{mn}$ support guesses only changes the constant in the base, and summing over the $D+1=O(\log(m/\eps))$ levels gives the expected query and runtime bounds stated in Theorem~\ref{thm:overview-main}.

\subsection{Further Related Work}\label{subsec:related}
Hirsch, Papadimitriou, and Vavasis~\cite{hirsch1989exponential} proved a $(1/\varepsilon)^{\Omega(n)}$ query lower bound for $O(1)$-Lipschitz maps on $[0,1]^n$ in the $\ell_\infty$ norm, which Chen and Deng~\cite{chen2008matching} later improved, proving an $O(L_1/\varepsilon)^{n-1}$ upper bound and a nearly matching lower bound. Rubinstein~\cite{rubinstein2017settling} proved a $2^{\Omega(n)}$ query lower bound in the relative $\ell_2$ norm when both the Lipschitz constant and the accuracy are constants.

Simplicial homotopies~\cite{eaves1972homotopies,merrill1972applications,van1982computation,van1987simplicial,bozantan2010implementation} trace a path through a triangulation of the domain by pivoting between adjacent simplices; they require no differentiability, but their step size is dictated by the mesh, leading to $(1/\varepsilon)^{\Omega(n)}$ worst-case complexity. For smooth homotopies, probability-one homotopy theory~\cite{chow1978finding} guarantees, via a parameterized Sard's theorem, that for almost every choice of parameters the path is a regular curve leading to a solution. The Stay-on-the-Ridge dynamics of~\cite{daskalakis2023stay} for min-max critical points in nonconvex-nonconcave games~\cite{daskalakis2021complexity} rely on more intricate second-order dynamics and come with no explicit runtime bound for their discretization.

In relation to algorithms for special classes of fixed-point and game-theoretic problems, Savani and von Stengel~\cite{savani2006hard} established an exponential (in the number of strategies) worst-case lower bound on the number of pivoting steps taken by the celebrated Lemke--Howson algorithm~\cite{lemke1964equilibrium} for two-player bimatrix games. While two-player zero-sum games are solvable in polynomial time via linear programming or no-regret learning, Adsul, Garg, Mehta, Sohoni, and von Stengel~\cite{AdsulGMSS21} also gave a polynomial-time algorithm for two-player rank-$1$ games. Their algorithm is inspired by path-following methods for fixed points, exploiting the low-rank structure of rank-$1$ games to achieve polynomial time---a structure that is quite delicate, as rank-$3$ games are already PPAD-hard~\cite{mehta2018constant}.

\paragraph{Acknowledgment.} We thank an anonymous reviewer of a previous version of this paper for proposing a recursive cover-based algorithm for fixed-point computation, which is essentially Algorithm~\ref{alg:tree} (Section~\ref{sec:fixed-point-bf}), and which we adapted to obtain our algorithm for Nash equilibrium computation (Section~\ref{sec:nash}).

\paragraph{Organization.} The remainder of the paper is organized as follows. In \Cref{sec:nash,sec:fixed-point-bf,sec:path-following}, we present and analyze our smoothed-analysis algorithms for normal-form Nash equilibria (via recursive grid refinement), fixed points and variational inequalities (via recursive brute force), and variational inequalities and fixed points (via homotopy path-following), respectively. In \Cref{sec:lowerBound}, we establish our query complexity lower bounds for fixed points and variational inequalities.

%% file: nash.tex
\section{Smoothed analysis of Finding an Approximate Nash}\label{sec:nash}

In this section, we prove Theorem~\ref{thm:overview-main} and then derive its corollary. Throughout this section, $\cC>1$ denotes a universal constant whose value may change from one occurrence to another.

\subsection{Proof of Theorem~\ref{thm:overview-main}}

We fix one support profile
\[
A=(A_1,\dots,A_m),\qquad \varnothing\neq A_i\subseteq[n],
\]
and search only over mixed profiles whose $i$th coordinate is supported on $A_i$.

\subsubsection{Notations}

For a nonempty set $A\subseteq[n]$, write
\[
\Delta(A):=\{x\in\Delta_n:x(a)=0\text{ for all }a\notin A\}.
\]
For the fixed support profile $A=(A_1,\dots,A_m)$, let
\[
K_A:=\Delta(A_1)\times\cdots\times\Delta(A_m),
\qquad
R:=\sum_{i=1}^m (|A_i|-1),
\]
and equip $K_A$ with the product $\ell_1$ metric
\[
\|x-y\|_1:=\sum_{i=1}^m \|x_i-y_i\|_1.
\]
For $c\in K_A$ and $r>0$, write
\[
B_A(c,r):=\{x\in K_A:\|x-c\|_1\le r\}.
\]
Lastly, we define the radius for searching
\[
\rho_t=2m\,2^{-t}.
\]

\subsubsection{The searching algorithm}

The algorithm proceeds as the pseudocode in Algorithm \ref{alg:main}. This algorithm proceeds in $D+1$ rounds, with $D=\left\lceil \log_2\frac{4m}{\eps}\right\rceil$. For each round $t$ of support guess $A$, it stores a frontier $\mathcal F_t(A)$ of centers waiting to be tested at round $t$. 
A center is kept only if the payoffs are consistent with $A$ being the support of an equilibrium; surviving centers are locally refined using the fixed grids below.
We call these centers \emph{level-$t$ candidates}, and each candidate $c$ represents the ball $B_A(c,\rho_t)$. 
To quantify this closeness to equilibrium, for each $c\in \mathcal F_t(A)$, we give a test as follows: 

A level-$t$ candidate center $c\in K_A$ \emph{passes the test} if for every player $i$ there exists an anchor action $\alpha_i\in A_i$ such that, at the mixed profile $x=c$,
\begin{equation}\label{eq:survival-support}
    U_{i,\alpha_i}(x) \le U_{i,a}(x) \le U_{i,\alpha_i}(x)+2\rho_t \qquad\text{for all }a\in A_i,
\end{equation}
\begin{equation}
\label{eq:survival-outside}
U_{i,b}(x) \le U_{i,\alpha_i}(x)+2\rho_t \qquad\text{for all }b\notin A_i.
\end{equation}
The first condition says that all actions in the guessed support are almost tied. The second says that actions outside the support are not substantially better.

Whenever a candidate center $c$ passes the test and $t<D$, the algorithm refines it using a fixed grid. At the first refinement ($t=0$, where $\mathcal F_0(A)$ consists of the single barycenter $b_A$), it generates the whole first-level grid $\mathcal N_1(A)$; at later refinements ($t\ge 1$), it generates only the children of $c$ in $\mathcal N_{t+1}(A)$. We use the following lemma to establish this process:

\begin{lemma}[Fixed grid refinement]\label{lem:subdivision-primitive}
Let $q_0=\max\{1,\lceil R/m\rceil\}$, $q_t=2^tq_0$, and
\[
\mathcal N_t(A):=\{x\in K_A:q_tx_i(a)\in\mathbb Z_{\ge0}\text{ for all }i\in[m],\ a\in A_i\}.
\]
For every $t\ge0$, this is a $\rho_t$-net of $K_A$ with $|\mathcal N_t(A)|\le\cC^{mn}2^{tR}$. For a point $c\in\mathcal N_t(A)$, its children are the feasible points of $\mathcal N_{t+1}(A)$ obtained by replacing each numerator $q_tc_i(a)$, for $a\in A_i\setminus\{\max A_i\}$, by either $2q_tc_i(a)$ or $2q_tc_i(a)+1$, and then setting the last coordinate of each player so that its coordinates sum to $1$. Each point has at most $2^R$ children, each child lies in $B_A(c,\rho_t)$, and every point of $\mathcal N_{t+1}(A)$ has a unique parent in $\mathcal N_t(A)$.
\end{lemma}

After all level-$t$ centers have been processed, the generated points become the next frontier. For $t\ge1$, uniqueness of the parent ensures that no point of $\mathcal N_{t+1}(A)$ is generated twice.

\begin{algorithm}[ht]
\caption{Search over support profiles and candidate centers}\label{alg:main}
\begin{algorithmic}[1]
\Require Realized game $G$, accuracy $\eps\in(0,1]$
\State $D\gets\left\lceil\log_2\frac{4m}{\eps}\right\rceil$
\ForAll{support profiles $A=(A_1,\dots,A_m)$ with $\varnothing\neq A_i\subseteq[n]$}
    \State let $b_A$ be the barycenter of $K_A$
    \State $\mathcal F_0(A)\gets\{b_A\}$
    \For{$t=0,1,\dots,D$}
        \If{$t<D$}
            \State $\mathcal F_{t+1}(A)\gets\varnothing$
        \EndIf
        \ForAll{$c\in\mathcal F_t(A)$}
            \State compute $\mathcal U(c)=\bigl(U_{i,a}(c)\bigr)_{i\in[m],a\in[n]}$
            \If{$c$ does not pass the level-$t$ test \eqref{eq:survival-support}--\eqref{eq:survival-outside}}
                \State \textbf{continue}
            \EndIf
            \If{$t=D$}
                \State \Return $c$
            \EndIf
            \If{$t=0$}
                \State add all points of $\mathcal N_1(A)$ to $\mathcal F_1(A)$
            \Else
                \State add all children of $c$ in \Cref{lem:subdivision-primitive} to $\mathcal F_{t+1}(A)$
            \EndIf
        \EndFor
    \EndFor
\EndFor
\State \Return \textsc{Fail}
\end{algorithmic}
\end{algorithm}

The pseudocode only assumes a subroutine for computing the full deviation-payoff vector at a chosen mixed profile. For the analysis, let $S_t(A)$ be the number of level-$t$ centers that pass the test and let $N_t(A):=|\mathcal F_t(A)|$ be the number of level-$t$ centers that would be tested if the algorithm continued after any earlier return.
Thus, the number of centers actually tested in one run is at most $\sum_A\sum_{t=0}^D N_t(A)$.
\begin{theorem}\label{thm:main-nash}
Assume \Cref{ass:density}, and suppose that for any mixed profile $x$ one can compute the full deviation-payoff vector $\mathcal U(x)$ in time $T$. Let $0<\eps\le1$. Then there exists a universal constant $\cC>1$ such that \Cref{alg:main} returns an $\eps$-approximate Nash equilibrium for every realized game.

Moreover, the expected number of payoff-vector computations is at most
\[
\Bigl(1+\Bigl\lceil \log_2\frac{4m}{\eps}\Bigr\rceil\Bigr)(\cC mM)^{mn}.
\]
Including the work of generating the grid points and their children, the expected running time is at most
\[
T\Bigl(1+\Bigl\lceil\log_2\frac{4m}{\eps}\Bigr\rceil\Bigr)(\cC mM)^{mn}.
\]
The algorithm is deterministic, and the expectation is over the realized game.
\end{theorem}

We first show that every returned point is correct, then prove that the fixed refinements keep an exact equilibrium covered through the final level, and finally bound the expected running time.

\subsubsection{Correctness of the algorithm}

We present several lemmas to demonstrate the correctness of the algorithm. 

\begin{lemma}[Lipschitzness of deviation payoffs]\label{lem:lipschitz}
For every player $i$, action $a\in[n]$, and mixed profiles $x,y\in(\Delta_n)^m$,
\[
|U_{i,a}(x)-U_{i,a}(y)| \le \sum_{j\neq i}\|x_j-y_j\|_1.
\]
In particular, if $c$ is a level-$t$ candidate and $y\in B_A(c,\rho_t)$, then
\[
|U_{i,a}(c)-U_{i,a}(y)| \le \rho_t=2m\,2^{-t}
\qquad\text{for every }i\in[m],\ a\in[n].
\]
\end{lemma}

\begin{proof}
Fix $i$ and $a$. Change the opponents' mixed strategies one player at a time and telescope. For a fixed opponent $j\neq i$, the map
\[
p\longmapsto u_i\bigl(a,(x_\ell)_{\ell<j,\,\ell\ne i},p,(y_\ell)_{\ell>j,\,\ell\ne i}\bigr)
\]
is affine in $p$ and takes values in $[0,1]$ (we used \Cref{asm:bounded-density games} so any pure strategy $u_i(a,x_{-i})\in[0,1]$), so changing $p$ from $x_j$ to $y_j$ changes the value by at most $\|x_j-y_j\|_1$. Summing over $j\neq i$ proves the first inequality.

For the second part, if $y\in B_A(c,\rho_t)$ then $\sum_{j=1}^m \|c_j-y_j\|_1\le \rho_t$, and hence
\[
\sum_{j\neq i}\|c_j-y_j\|_1\le \rho_t.
\]
Applying the first inequality with $x=c$ proves the claim.
\end{proof}

\begin{lemma}[A candidate that passes the test is already approximate]\label{lem:pass-implies-ane}
If a level-$t$ candidate center $c$ passes the test, then $c$ is a $2\rho_t$-approximate Nash equilibrium.
\end{lemma}

\begin{proof}
Let $x=c$ and fix a player $i$. Since $c$ passes the test, there is an anchor action $\alpha_i\in A_i$ such that
\[
U_{i,c}(x)\le U_{i,\alpha_i}(x)+2\rho_t
\qquad\text{for every }c\in[n].
\]
On the other hand, every action in the support of $x_i$ lies in $A_i$, and \eqref{eq:survival-support} gives
\[
U_{i,a}(x)\ge U_{i,\alpha_i}(x)
\qquad\text{for every }a\in A_i.
\]
Since $u_i(x)$ is a convex combination of the values $U_{i,a}(x)$ over actions $a$ used by $x_i$, we obtain
\[
u_i(x)\ge U_{i,\alpha_i}(x).
\]
Therefore
\[
\max_{c\in[n]}U_{i,c}(x)-u_i(x)\le 2\rho_t.
\]
This holds for every player $i$, so $x$ is a $2\rho_t$-Nash equilibrium.
\end{proof}

\begin{lemma}[A center close to the exact equilibrium survives]\label{lem:near-nash-survives}
Let $x^\star$ be an exact Nash equilibrium of the realized game with support profile $A$. If $c\in K_A$ satisfies $\|c-x^\star\|_1\le \rho_t$, then $c$ passes the level-$t$ test.
\end{lemma}

\begin{proof}
Fix a player $i$. Because $x^\star$ is an exact Nash equilibrium and $A_i$ is exactly the set of actions played with positive probability by player $i$, all actions in $A_i$ have the same deviation payoff against $x^\star_{-i}$; call this common value $v_i$. Every action outside $A_i$ has deviation payoff at most $v_i$.

By \Cref{lem:lipschitz}, for every action $a\in[n]$,
\[
|U_{i,a}(c)-U_{i,a}(x^\star)|\le \rho_t.
\]
Choose $\alpha_i\in A_i$ minimizing $U_{i,a}(c)$ over $a\in A_i$. Then
\[
U_{i,\alpha_i}(c)\ge v_i-\rho_t.
\]
For every $a\in A_i$,
\[
U_{i,a}(c)\le v_i+\rho_t\le U_{i,\alpha_i}(c)+2\rho_t,
\]
and by the choice of $\alpha_i$ we also have $U_{i,\alpha_i}(c)\le U_{i,a}(c)$. Likewise, for every $b\notin A_i$,
\[
U_{i,b}(c)\le v_i+\rho_t\le U_{i,\alpha_i}(c)+2\rho_t.
\]
Thus $c$ passes the level-$t$ test.
\end{proof}

\begin{lemma}[The exact equilibrium remains covered]\label{lem:surviving-path}
Let $x^\star$ be an exact Nash equilibrium of the realized game, and let $A$ be its support profile. In \Cref{alg:main}, unless a point has already been returned, for every $t=0,1,\dots,D$ there exists a center $c_t\in\mathcal F_t(A)$ such that
\[
x^\star\in B_A(c_t,\rho_t).
\]
\end{lemma}

\begin{proof}
We use induction on $t$. The base case $t=0$ is immediate from $K_A\subseteq B_A(b_A,\rho_0)$. We can see it as follows: Let $b_A\in K_A$ be the barycenter of $K_A$, namely the product of the uniform distributions on $A_1,\dots,A_m$. Since each $\Delta(A_i)$ has $\ell_1$-diameter at most $2$, the product region $K_A$ has $\ell_1$-diameter at most $2m$, and hence $K_A\subseteq B_A(b_A,\rho_0)$.

For $t\ge1$, define $c_t\in\mathcal N_t(A)$ by rounding each coordinate of $x_i^\star$ except $a=\max A_i$ down to $\lfloor q_tx_i^\star(a)\rfloor/q_t$, and then setting the last coordinate $c_{t,i}(\max A_i)=1-\sum_{a\in A_i\setminus\{\max A_i\}}c_{t,i}(a)$ so that the coordinates sum to $1$. Since $0\le x_i^\star(a)-c_{t,i}(a)<1/q_t$ for each of the $R=\sum_{i=1}^m(|A_i|-1)$ non-maximal coordinates, we have $\|c_t-x^\star\|_1=2\sum_{i=1}^m\sum_{a\in A_i\setminus\{\max A_i\}}(x_i^\star(a)-c_{t,i}(a))\le 2R/q_t\le\rho_t$ (as in the proof of \Cref{lem:subdivision-primitive}). Moreover, since $q_{t+1}=2q_t$, for every $a\in A_i\setminus\{\max A_i\}$ we have
\[
\lfloor q_{t+1}x_i^\star(a)\rfloor\in\bigl\{2\lfloor q_tx_i^\star(a)\rfloor,\ 2\lfloor q_tx_i^\star(a)\rfloor+1\bigr\},
\]
so $c_{t+1}$ is a child of $c_t$ in \Cref{lem:subdivision-primitive}. By \Cref{lem:near-nash-survives}, $b_A$ passes the level-$0$ test, so the algorithm generates all of $\mathcal N_1(A)$, including $c_1\in\mathcal F_1(A)$. Inductively, for every $1\le t<D$, since $c_t\in\mathcal F_t(A)$ satisfies $\|c_t-x^\star\|_1\le\rho_t$, \Cref{lem:near-nash-survives} implies that $c_t$ passes the level-$t$ test and hence generates its child $c_{t+1}\in\mathcal F_{t+1}(A)$. This completes the induction.
\end{proof}

\begin{lemma}[Correctness of any returned output]\label{lem:output-correct}
Whenever \Cref{alg:main} returns a point, that point is an $\eps$-approximate Nash equilibrium. Moreover, for every realized game, \Cref{alg:main} returns a point.
\end{lemma}

\begin{proof}
If the algorithm returns, it returns a level-$D$ center that passed the level-$D$ test. By \Cref{lem:pass-implies-ane}, every such center is a $2\rho_D$-approximate Nash equilibrium. Since
\[
2\rho_D=4m\,2^{-D}\le \eps,
\]
the returned profile is an $\eps$-approximate Nash equilibrium.

It remains to see that the algorithm returns a point. Let $x^\star$ be an exact Nash equilibrium of the realized game, and let $A$ be its support profile. By \Cref{lem:surviving-path}, unless a point has already been returned, the frontier $\mathcal F_D(A)$ contains a center $c_D$ with $\|c_D-x^\star\|_1\le\rho_D$. By \Cref{lem:near-nash-survives}, this center passes the level-$D$ test. Therefore, when the algorithm reaches this support guess at level $D$, it returns a level-$D$ center that passes the test.
\end{proof}

\subsubsection{Bounding the number of surviving candidates}

Here we recall the notations. For a fixed support profile $A=(A_1,\dots,A_m)$, write
\[
r_i:=|A_i|,
\qquad
R:=\sum_{i=1}^m(r_i-1).
\]
For $\tau\ge 0$, let $E_A(\tau)$ be the set of points $x\in K_A$ for which, for every player $i$, there exists an anchor $\alpha_i\in A_i$ satisfying
\[
\begin{aligned}
U_{i,\alpha_i}(x) &\le U_{i,a}(x) \le U_{i,\alpha_i}(x)+\tau &&\text{for all }a\in A_i,\\
U_{i,b}(x) &\le U_{i,\alpha_i}(x)+\tau &&\text{for all }b\notin A_i.
\end{aligned}
\]
Thus a level-$t$ center passes the test exactly when it belongs to $E_A(2\rho_t)$.

\begin{lemma}[Probability of passing the relaxed survival test]\label{lem:volume-bound}
For every fixed $x\in K_A$ and every $\tau\ge0$,
\[
\mathbb P[x\in E_A(\tau)]\le M^{mn}(2\tau)^R.
\]
\end{lemma}

\begin{proof}
Fix $x\in K_A$. For a player $i$, let $S_i(\tau)\subseteq [0,1]^n$ be the set of vectors $v=(v_a)_{a\in[n]}$ for which there exists an anchor $\alpha\in A_i$ satisfying
\[
\begin{aligned}
 v_{\alpha} &\le v_a \le v_{\alpha}+\tau &&\text{for all }a\in A_i,\\
 v_b &\le v_{\alpha}+\tau &&\text{for all }b\notin A_i.
\end{aligned}
\]
For a fixed anchor $\alpha$, integrate first over $v_{\alpha}\in[0,1]$. Each of the remaining $r_i-1$ support coordinates is then confined to an interval of length at most $\tau$, while the other coordinates contribute a factor at most $1$. Therefore the corresponding volume is at most $\tau^{r_i-1}$. Summing over the $r_i$ choices of $\alpha$ gives
\[
\lambda_n\bigl(S_i(\tau)\bigr)\le r_i\tau^{r_i-1}.
\]
Hence
\[
\mathbb P[x\in E_A(\tau)]
\le M^{mn}\prod_{i=1}^m r_i\tau^{r_i-1}
\le M^{mn}(2\tau)^R,
\]
where we used $r_i\le 2^{r_i-1}$ for every $i$. This proves the claim.
\end{proof}

\begin{lemma}[Stability of the relaxed survival region]\label{lem:relaxed-stability}
If $x\in E_A(\tau)$, $y\in K_A$, and $\|x-y\|_1\le \delta$, then $y\in E_A(\tau+2\delta)$.
\end{lemma}

\begin{proof}
Fix a player $i$ and let $\alpha_i\in A_i$ be an anchor certifying the membership of $x$ in $E_A(\tau)$. Also let $\beta_i\in A_i$ minimize $U_{i,a}(y)$ over $a\in A_i$. By \Cref{lem:lipschitz}, every deviation payoff changes by at most $\delta$ when moving from $x$ to $y$. Therefore, for every $a\in A_i$,
\[
U_{i,a}(y)\le U_{i,a}(x)+\delta\le U_{i,\alpha_i}(x)+\tau+\delta.
\]
On the other hand,
\[
U_{i,\beta_i}(y)\ge U_{i,\beta_i}(x)-\delta\ge U_{i,\alpha_i}(x)-\delta,
\]
because $U_{i,\alpha_i}(x)\le U_{i,a}(x)$ for all $a\in A_i$. Combining the last two displays gives
\[
U_{i,a}(y)\le U_{i,\beta_i}(y)+\tau+2\delta
\qquad\text{for every }a\in A_i,
\]
and the lower bound $U_{i,\beta_i}(y)\le U_{i,a}(y)$ holds by the choice of $\beta_i$. The same argument gives
\[
U_{i,b}(y)\le U_{i,\beta_i}(y)+\tau+2\delta
\qquad\text{for every }b\notin A_i.
\]
Thus the anchors $\beta_i$ certify that $y\in E_A(\tau+2\delta)$.
\end{proof}

\begin{lemma}[Expected number of tested centers at a fixed level]\label{lem:count-pass}
Fix a support profile $A=(A_1,\dots,A_m)$. Then $N_0(A)=1$, and for every $1\le t\le D$,
\[
\mathbb E[N_t(A)]\le (\cC mM)^{mn}.
\]
\end{lemma}

\begin{proof}
Fix $t\ge 1$. Every center $c\in\mathcal F_t(A)$ is generated from a level-$(t-1)$ center that passed the test. For $t\ge2$, \Cref{lem:subdivision-primitive} places every child within distance $\rho_{t-1}$ of its parent; for $t=1$, every point of $\mathcal N_1(A)$ lies within distance $\rho_0=2m$ of $b_A$, because $K_A$ has diameter at most $2m$. As a passing level-$(t-1)$ center belongs to $E_A(2\rho_{t-1})$, \Cref{lem:relaxed-stability} gives
\[
\mathcal F_t(A)\subseteq E_A(4\rho_{t-1})=E_A(8\rho_t).
\]
Moreover, every tested point at level $t\ge 1$ belongs to the fixed grid $\mathcal N_t(A)$ (with no duplicates, by uniqueness of parents in \Cref{lem:subdivision-primitive}), so $\mathcal F_t(A)\subseteq \mathcal N_t(A)\cap E_A(8\rho_t)$. Thus
\[
N_t(A)\le\sum_{x\in\mathcal N_t(A)}\mathbf1_{\{x\in E_A(8\rho_t)\}}.
\]
Taking expectations and applying \Cref{lem:volume-bound} at each fixed grid point $x\in\mathcal N_t(A)$, followed by the cardinality bound $|\mathcal N_t(A)|\le\cC^{mn}2^{tR}$ from \Cref{lem:subdivision-primitive}, gives
\[
\begin{aligned}
\mathbb E[N_t(A)]&\le|\mathcal N_t(A)|M^{mn}(16\rho_t)^R\\
&\le\cC^{mn}2^{tR}M^{mn}(32m\,2^{-t})^R\le(\cC mM)^{mn},
\end{aligned}
\]
after enlarging $\cC$. This proves the claim.
\end{proof}

\begin{proof}[of \Cref{thm:main-nash}]
By \Cref{lem:output-correct}, \Cref{alg:main} returns an $\eps$-approximate Nash equilibrium for every realized game.

It remains to bound the expected number of tested candidates. For a fixed support profile $A$, we have $N_0(A)=1$ and, by \Cref{lem:count-pass},
\[
\mathbb E[N_t(A)]\le (\cC mM)^{mn}
\qquad (1\le t\le D).
\]
There are at most $(2^n-1)^m\le 2^{mn}$ support profiles. Since $R\le mn$ for every support profile, after summing over all support guesses the expected number of tested candidates at each level is at most
\[
2^{mn}(\cC mM)^{mn}\le(\cC mM)^{mn},
\]
after enlarging $\cC$. The same bound also absorbs the roots. Summing over the $D+1$ levels gives
\[
\mathbb E[\#\text{ payoff-vector computations}]
\le (D+1)(\cC mM)^{mn}
\]
after enlarging $\cC$.

For the running time, each tested candidate requires one computation of $\mathcal U(x)$, which costs time $T$, plus $O(mn)$ operations to check the inequalities in \eqref{eq:survival-support}--\eqref{eq:survival-outside}. At $t=0$, if $b_A$ passes the test, the first refinement enumerates $\mathcal N_1(A)$ by listing the nonnegative integer vectors summing to $q_1$ for each player; this takes $\mathrm{poly}(mn)$ coordinate operations per generated point, and every generated point is placed in $\mathcal F_1(A)$ (with $|\mathcal N_1(A)|\le\cC^{mn}2^R$ by \Cref{lem:subdivision-primitive}). At every later refinement ($1\le t<D$), each passing center generates at most $2^R$ candidate children, taking $\mathrm{poly}(mn)$ coordinate operations per child. Since every child in $\mathcal N_{t+1}(A)$ has a unique parent in $\mathcal N_t(A)$, these lists contain no duplicates and require no packing step.

Using $S_t(A)\le N_t(A)$, the expected-count bound above, and $R\le mn$ (and noting that computing $\mathcal U(x)\in[0,1]^{mn}$ takes time $T=\Omega(mn)$), the exponential factor $2^R$ and the $\mathrm{poly}(mn)$ coordinate overhead per candidate are absorbed by enlarging $\cC$. Thus, the expected running time is at most
\[
T(D+1)(\cC mM)^{mn}.
\]
\end{proof}

\subsubsection{Proof of the fixed-grid lemma}\label{sec:subdivision-proof}

We now prove \Cref{lem:subdivision-primitive}. The grids depend only on the support profile and the level, not on the realized game.

\begin{proof}[of \Cref{lem:subdivision-primitive}]
If $R=0$, every $A_i$ contains one action, so $K_A$ and every grid consist of one point, and all claims are immediate. Suppose $R>0$. For $x\in K_A$, round each coordinate except $a=\max A_i$ down to $c_i(a)=\lfloor q_tx_i(a)\rfloor/q_t$, and set $c_i(\max A_i)=1-\sum_{a\in A_i\setminus\{\max A_i\}}c_i(a)$. The last coordinate is nonnegative, so $c\in\mathcal N_t(A)$. The errors in the other coordinates are nonnegative and are added to the last coordinate, hence
\[
\begin{aligned}
\|x-c\|_1&=2\sum_{i=1}^m\sum_{a\in A_i\setminus\{\max A_i\}}(x_i(a)-c_i(a))\\
&\le\frac{2R}{q_t}\le2m\,2^{-t}=\rho_t.
\end{aligned}
\]
Thus $\mathcal N_t(A)$ is a $\rho_t$-net.

Recall that $r_i=|A_i|$ and $R=\sum_i(r_i-1)$. Counting nonnegative integer vectors with sum $q_t$ gives
\[
|\mathcal N_t(A)|=\prod_{i=1}^m\binom{q_t+r_i-1}{r_i-1}\le\binom{mq_t+R}{R}.
\]
For the inequality, partition a set of size $mq_t+R$ into sets of sizes $q_t+r_i-1$: choosing $r_i-1$ elements from each is one of the ways to choose $R$ elements in total. Since $mq_0\le R+m$, we have $mq_t+R\le(2R+m)2^t$, and consequently
\[
\begin{aligned}
|\mathcal N_t(A)|&\le\left(\frac{e(mq_t+R)}{R}\right)^R\\
&\le(2e)^R\left(1+\frac{m}{2R}\right)^R2^{tR}\le(2e)^Re^{m/2}2^{tR}\le\cC^{mn}2^{tR},
\end{aligned}
\]
where the last inequality uses $R+m=\sum_i|A_i|\le mn$.

There are $R$ independent numerator choices, each with two possibilities, so a point has at most $2^R$ children. Each independent coordinate increases by either $0$ or $1/q_{t+1}$, with the increase subtracted from the last coordinate of that player. Every feasible child is therefore within distance $2R/q_{t+1}=R/q_t\le\rho_t$ of its parent. Conversely, for $y\in\mathcal N_{t+1}(A)$, the parent must have independent numerators $\lfloor q_ty_i(a)\rfloor$. These determine a unique feasible parent, since their sum for each player is at most $q_t$. This proves the claim.
\end{proof}

The following corollary gives a concrete perturbation model based on random action costs that satisfies \Cref{ass:density}, so that \Cref{thm:main-nash} applies.

\begin{corollary}[random action costs]\label{cor:costs}
Fix any baseline game $\bar u_i:[n]^m\to[0,1]$. For each player-action pair $(i,a)$ let $c_{i,a}$ be a random cost, and define the perturbed game by
\[
u_i(a_i,a_{-i}) := \bar u_i(a_i,a_{-i}) - c_{i,a_i}.
\]
As in the setting above, assume that the perturbed payoffs lie in $[0,1]$ almost surely.
Assume either of the following.
\begin{enumerate}[label=(\roman*)]
\item The random variables $\{c_{i,a}: i\in[m],\ a\in[n]\}$ are independent and each has a density bounded by $M$.
\item The whole cost vector $C=(c_{i,a})_{i,a}\in\R^{mn}$ has a joint density bounded by $M^{mn}$.
\end{enumerate}
Then \Cref{ass:density} holds with the same $M$. Consequently, \Cref{thm:main-nash} applies.
\end{corollary}

\begin{proof}
Fix any mixed profile $x$. Let
\[
\bar U_{i,a}(x):=\bar u_i(a,x_{-i})
\]
be the deviation payoffs in the baseline game. In the perturbed game,
\[
U_{i,a}(x)=\bar U_{i,a}(x)-c_{i,a},
\]
so the vector $(U_{i,a}(x))_{i,a}$ is obtained from the deterministic vector $(\bar U_{i,a}(x))_{i,a}$ by translation by $-C$.

In case (i), the joint density of $C$ is the product of the marginal densities, hence is bounded by $M^{mn}$. In case (ii), this is assumed directly. Translation preserves density bounds, so the same bound holds for the deviation-payoff vector at every fixed profile $x$. This is exactly \Cref{ass:density}.
\end{proof}

%% file: rec-fixed.tex
\providecommand{\DIFdel}[1]{{\color{red}#1}}
\providecommand{\DIFadd}[1]{{\color{blue}#1}}

\section{Smoothed Analysis of Fixed-Point Computation via Recursive Brute Force}\label{sec:fixed-point-bf}

We are going to find a fixed point of a function $F:\mathcal B(0,1)\to\mathcal B(0,1)$. We recall these assumptions for our use:

\begin{assumption}[Lipschitzness, Assumption \ref{asm:F-Lip}]
   $F$ is $L_1$-Lipschitz: $\forall x,y\in\cB(0,1),\ \|F(x)-F(y)\|\le L_1\|x-y\|$.
\end{assumption}

\begin{assumption}[Bounded density, Assumption \ref{asm:bounded-density}]
    There exists $P_v$ as the upper bound of the density of the value of $F(x)$, for all $x$ in the region and for all possible values of $F(x)$.
\end{assumption}

\subsection{Recursive Brute Force Algorithm}

Now we use the following algorithm for the fixed-point computation:

\begin{algorithm}[h]
\caption{Finding Fixed Point by Tree Search and Pruning} 
\label{alg:tree}
\begin{algorithmic}[1] 
\Require Function $F$, $0<\varepsilon\le1$
\State Set $K \gets \left\lceil \log_2\left(\frac{1+L_1}{\varepsilon}\right)\right\rceil$, $D_0\gets \{0\}$.
\For{$k=1,\dots,K$}\Comment{Construct a tree}
    \State $D_k\gets\emptyset, r_k=2^{-k}$.
    \State Define implicitly an $r_k$-covering $C_k=\{(x_1,\dots,x_n)\in\cB(0,1)|\forall i,~x_i\cdot\sqrt{n}/r_k\in\mathbb Z\}$. 
    \For{$x\in D_{k-1}$}
        \State Choose a set $C(x,k)=C_k\cap \cB(x,3r_k)$
        \State $D_k$ add all points $y\in C(x,k)$ such that $\|F(y)-y\|\le r_k\cdot (1+L_1)$.
    \EndFor
\EndFor
\Ensure Any point in $D_K$.
\end{algorithmic}
\end{algorithm}

This algorithm will yield the following bound:

\begin{theorem}\label{thm:brute-force}
    Suppose we perform Algorithm \ref{alg:tree}. Let $0<\varepsilon\le1$. Then, for a function $F:\cB (0,1)\to\cB(0,1)$ satisfying the Assumptions \ref{asm:F-Lip} and \ref{asm:bounded-density}, the algorithm finds an approximate fixed point that satisfies $\|F(x)-x\|\le\varepsilon$ using an expected number of evaluations of $F$ at most
    \[(135(1+L_1))^n\cdot P_v\cdot \mathrm{Vol}(\mathcal B_n)\cdot \left\lceil \log_2\left(\frac{1+L_1}{\varepsilon}\right)\right\rceil,\]
    where $\mathrm{Vol}(\mathcal B_n)$ is the volume of an $n$-dimensional ball. Its expected running time is at most $\mathcal C^n$ times this bound, for an absolute constant $\mathcal C$ (Remark~\ref{rem:enumeration}).
\end{theorem}

Notice that we can use the following folklore lemma for the number of branches:

\begin{lemma}[Covering]\label{lem:covering}
    Let $0<\varepsilon<1$. Then, we can use $(5(1+\varepsilon/2)/\varepsilon)^n$  $n$-dimensional balls of radius $\varepsilon$ to cover a  $n$-dimensional ball of radius $1$.
\end{lemma}

\begin{proof}
    Consider a grid 
    \[C=\{(x_1,\dots,x_n)|\forall i,~x_i\cdot\sqrt{n}/\varepsilon\in\mathbb Z \wedge x_1^2+\dots+x_n^2\le 1\}.\]
    Then, the set union of $\mathcal B(x,\varepsilon)$ for $x\in C$ covers $\mathcal B(0,1)$.
Indeed, rounding each coordinate toward zero to a multiple of $\varepsilon/\sqrt n$ gives a point in $C$ within distance $\varepsilon$. Notice that for $x\in C$, by the construction, the interiors of the $L_\infty$ hypercubes $\{t|\|t-x\|_\infty\le\varepsilon/(2\sqrt{n})\}$ are disjoint. The union of these hypercubes is a subset of a larger ball $\cB(0,1+\varepsilon/2)$. Therefore, by a volume argument, 
    \[|C|\le \frac{\mathrm{Vol}(\cB_n)}{(\varepsilon/\sqrt{n})^n}(1+\varepsilon/2)^n\le (e\sqrt{\pi}/\varepsilon)^n(1+\varepsilon/2)^n\le (5(1+\varepsilon/2)/\varepsilon)^n.\]
\end{proof}

Now we go into the main proof.

\begin{proof}
    We first prove the correctness of the algorithm. We should know that from Brouwer's Fixed Point theorem, there is going to be a fixed point of $F$. Therefore, we should know that for some fixed point $x^*$, any $y$ satisfy 
    \[\|F(y)-y\|\le \|F(y)-F(x^*)\|+\|F(x^*)-y\|= \|F(y)-F(x^*)\|+\|x^*-y\|\le (1+L_1)\|x^*-y\|.\]

    Therefore, if $x^*$ lies in some $\mathcal B(y,r)$, we have $\|F(y)-y\|\le r\cdot (1+L_1)$. By induction, we can say that all the $D_k$'s are not empty. The induction keeps a point of $D_k$ within distance $r_k$ of $x^*$: starting with $D_0=\{0\}$ and $r_0=1$, a covering point within $r_k$ of $x^*$ is within $r_{k-1}+r_k=3r_k$ of the retained point in $D_{k-1}$, and it passes the test above. Therefore, this algorithm will always terminate and find a point that satisfies $\|F(x)-x\|\le(1+L_1)\cdot r_K$. And by our construction, $r_K\le\frac{\varepsilon}{1+L_1}$. So we can find an approximate fixed point.

    Now we analyze the runtime of the algorithm. Let $K=\left\lceil \log_2\left(\frac{1+L_1}{\varepsilon}\right)\right\rceil$. For each $1\le k\le K$, the cover $|C_k|\le (5(1+r_k/2)/r_k)^n$. By Assumption \ref{asm:bounded-density}, $|D_k|$ is upper bounded by $|L_k|$ where $L_k$ is the set of points $x\in C_k$ such that $\|F(x)-x\|\le r_k\cdot(1+L_1)$. So, the expected number is
    \[\mathbb E[|L_k|]\le P_v\cdot \mathrm{Vol}(\mathcal B_n)\cdot ((1+L_1)\cdot r_k)^n\cdot (5(1+r_k/2)/r_k)^n\le P_v\cdot \mathrm{Vol}(\mathcal B_n)\cdot (\frac{15}{2}(1+L_1))^n\]


Also, we know that the union of $\mathcal B(y,r_k)$ for $y\in C(x,k)$ covers $\mathcal B(x,2r_k)\cap\mathcal B(0,1)$, because it takes all the points $y$ in $C_k$ that $\mathcal B(y,r_k)$ intersects $B(x,2r_k)\cap\mathcal B(0,1)$. By the construction, since $x\in C_{k-1}$, all the coordinate is an integer multiple of $r_{k-1}/\sqrt{n}$ hence is an integer multiple of $r_{k}/\sqrt{n}$ (because $r_{k-1}=2r_k$).  
Therefore, by a similar argument of Lemma \ref{lem:covering}, the grid cubes of side $r_k/\sqrt n$ centered at points of $C(x,k)$ have disjoint interiors and lie in $\cB(x,7r_k/2)$, so
\[
|C(x,k)|\le\frac{\mathrm{Vol}(\cB_n)(7r_k/2)^n}{(r_k/\sqrt n)^n}
\le(35/2)^n<18^n.
\] Therefore, for each round, the expected total number of value checked is at most

\[
\begin{aligned}
\mathbb E\!\left[\sum_{x\in D_{k-1}}|C(x,k)|\right]\le18^n\mathbb E[|D_{k-1}|]\le P_v\cdot\mathrm{Vol}(\mathcal B_n)\cdot(135(1+L_1))^n
\end{aligned}
\]
For $k=1$, the same bound applies because $|D_0|=1$ and $P_v\mathrm{Vol}(\cB_n)\ge1$, as $F(0)$ is supported in $\cB(0,1)$.
Therefore, the expected total number of evaluations of $F$ is at most  
\[P_v\cdot \mathrm{Vol}(\mathcal B_n)\cdot (135(1+L_1))^n\cdot K.\]
This concludes the theorem.
\end{proof}

\begin{remark}\label{rem:enumeration}
    In Algorithm \ref{alg:tree}, we use this covering to be the pool of candidate points. Enumerating such a pool needs $O(1/r_k)^n$ time. However, we do not go over the entire pool and calculate the intersection.
    In the proof in \ref{thm:brute-force}, what we do is that for every center $x$, to make the balls $\cB(y,r_k)$ cover $\cB(x,2r_k)\cap\cB(0,1)$, we only enumerate $C(x,k)=C_k\cap\cB(x,3r_k)$.
    Indeed, for every $z\in\cB(x,2r_k)\cap\cB(0,1)$, the covering $C_k$ gives a point $y\in C_k$ with $\|y-z\|\le r_k$, and thus $\|y-x\|\le\|y-z\|+\|z-x\|\le r_k+2r_k=3r_k.$
    Also all the $x$'s are grid points so all the constructed points are grid points. Therefore we can still use these centers.

    \emph{Running time.} Besides evaluating $F$, the algorithm enumerates each $C(x,k)$ and stores $D_k$. Write $y=x+(r_k/\sqrt n)\,j$ with $j\in\mathbb Z^n$, so $y\in\cB(x,3r_k)$ if and only if $\|j\|\le3\sqrt n$. We enumerate $j$ coordinate by coordinate, keeping only prefixes $(j_1,\dots,j_d)$ with $j_1^2+\dots+j_d^2\le 9n$, and discard the points outside $\cB(0,1)$. The prefixes of length $d$ are integer points of a $d$-dimensional ball of radius $3\sqrt n$, so by the volume argument of Lemma~\ref{lem:covering} there are at most $\mathrm{Vol}(\cB_d)(3\sqrt n+\sqrt d/2)^d\le(\mathcal C\sqrt{n/d})^d\le\mathcal C'^{\,n}$ of them, where the last step uses $\max_{d>0}(a\sqrt{n/d})^d=e^{a^2n/(2e)}$. Hence each $C(x,k)$ is enumerated in time $\mathcal C^n$. We store $D_k$ in a dictionary keyed by the integer coordinates of its points, so points generated from different parents are kept once, at cost $\mathrm{poly}(n)$ per insertion. Therefore, in Theorem~\ref{thm:brute-force} and Corollaries~\ref{cor:with_boundary} and~\ref{cor:reverse-vi-fp}, the expected running time is at most $\mathcal C^n$ times the stated expected number of evaluations.
\end{remark}

We admit we could take some perturbation that makes the function take the value outside of the boundary $\cB(0,1)$. Therefore, we include the version that allows $F$ to take a value and then project back to the ball. Hence, the density needs to be defined in another way.

\begin{corollary}\label{cor:with_boundary}
    Suppose we perform Algorithm \ref{alg:tree}. Let $0<\varepsilon\le1$. Let $F$ be a function from $\cB(0,1)\to\cB(0,1)$ satisfying Assumption \ref{asm:F-bounded}, the Lipschitz inequality in Assumption \ref{asm:F-Lip}, and Assumption \ref{asm:bounded-density} for all $F(x)=y$ such that $\|y\|<1$. Furthermore, for $\|y\|=1$, the $(n-1)$-dimensional density is at most $P_v'$ (i.e. for any $S\subseteq \mathbb S^{n-1}$, the probability is at most $\mathrm{Area}(S)\cdot P_v'$). Then the algorithm finds an approximate fixed point that satisfies $\|F(x)-x\|\le\varepsilon$ using an expected number of evaluations of $F$ at most
    \[(540(1+L_1))^n\cdot ( P_v+nP_v')\cdot \mathrm{Vol}(\mathcal B_n)\cdot \left\lceil \log_2\left(\frac{1+L_1}{\varepsilon}\right)\right\rceil,\]
    where $\mathrm{Vol}(\mathcal B_n)$ is the volume of an $n$-dimensional ball.
\end{corollary}

\begin{proof}
    We know that for $x\in C_k\subseteq\cB(0,1)$ the area of the $S$ on $\mathbb S^{n-1}$ that is having distance within $(1+L_1)r_k$ is at most $2^n\cdot\mathrm{Vol}(\cB_n)((1+L_1)r_k)^{n-1}$. Also, we can notice that the number of $x\in C_k$ that $\cB(x,(1+L_1)r_k)$ intersects the boundary, by our construction is at most
    \[
        \frac{n\,\mathrm{Vol}(\cB_n)(2+L_1)r_k(1+r_k/2)^{n-1}}
             {(r_k/\sqrt n)^n}
        \le n\cdot(15/2)^n(2+L_1)r_k^{-(n-1)}.
    \]
    This is because, for each center counted above, consider the grid cube of side $r_k/\sqrt n$ centered at $x$. These cubes have disjoint interiors and half-diagonal $r_k/2$. All these cubes are inside a shell. The outer boundary is these cubes are smaller than $1+r_k/2$, since $1-(1+L_1)r_k\le\|x\|\le1$, and the distance to the center with the boundary is at most $r_k/2$. Also, $x$ has norm at least $1-(1+L_1)r_k$, and the inner boundary of the cubes has norm at least $1-(L_1+3/2)r_k$. Thus, every point in their union has norm between $1-(L_1+3/2)r_k$ and $1+r_k/2$. The outer area of the shell is $n\mathrm{Vol}(\cB_n)(1+r_k/2)^{n-1}$, and the thickness is $(2+L_1)r_k$. Therefore, the volume of the shell, which is an upper bound of the union of the cube, has volume at most $n\mathrm{Vol}(\cB_n)(1+r_k/2)^{n-1}(2+L_1)r_k$. Then we can divide by the volume $(r_k/\sqrt n)^n$ of each cube gives the bound above.

    Therefore we can have the bound for $\mathbb E[|L_k|]$ as $\mathbb E[|L_{in}|]+\mathbb E[|L_{out}|]$, where $L_{in}=\{x\in L_k:\|F(x)\|<1\}$ is for the values inside $\cB(0,1)$, and for the points on the boundary, with $L_{out}=L_k\setminus L_{in}$, we have
\[
\begin{aligned}
\mathbb E[|L_{out}|]
&\le n\cdot(15/2)^n(2+L_1)r_k^{-(n-1)}\cdot2^n\cdot\mathrm{Vol}(\cB_n)((1+L_1)r_k)^{n-1}\cdot P_v'\\
&\le30^n(1+L_1)^n\cdot\mathrm{Vol}(\cB_n)\cdot nP_v'.
\end{aligned}
\]

   So the total $\mathbb E[|L_k|]$ is at most
\[
30^n(1+L_1)^n\cdot(P_v+nP_v')\cdot\mathrm{Vol}(\cB_n).
\]
Also, $(P_v+nP_v')\mathrm{Vol}(\cB_n)\ge1$, so the same bound applies to $D_0=\{0\}$. Combining this with $|C(x,k)|<18^n$ gives the stated bound, since $30^n18^n=540^n$. The rest of the proof goes the same as before.
\end{proof}
Finally, we use the above corollary to reduce the variational inequality to the fixed-point problem

\begin{corollary}\label{cor:reverse-vi-fp}
    Suppose we perform Algorithm \ref{alg:tree} on $G(x)=\Pi(x+F(x))$ with accuracy $\varepsilon/(2+L_0)$, where $\Pi$ is the projection to $\cB(0,1)$. Let $0<\varepsilon\le1$. Let $F$ be a function from $\cB(0,1)\to\mathbb R^n$ satisfying the Assumptions \ref{asm:F-bounded}, \ref{asm:F-Lip} and \ref{asm:bounded-density}. Then, the algorithm finds an approximate solution to the variational inequality that satisfies $\langle F(x),y-x\rangle\le\varepsilon,~\forall y\in\cB(0,1)$ using an expected number of evaluations of $F$ at most
    \[(700(2+L_1))^n\cdot (1+(1+L_0)^n)\cdot P_v\cdot \mathrm{Vol}(\mathcal B_n)\cdot \left\lceil \log_2\left(\frac{(2+L_1)(2+L_0)}{\varepsilon}\right)\right\rceil,\]
    where $\mathrm{Vol}(\mathcal B_n)$ is the volume of an $n$-dimensional ball.
\end{corollary}

\begin{proof}
    Denote $\Pi$ be the projection to the ball such that $\Pi(x)=x$ for $\|x\|\le 1$ and $\Pi(x)=x/\|x\|$ for $\|x\|>1$. Let $G(x)=\Pi(x+F(x))$. We claim that any $\varepsilon'=\eps/(2+L_0)$ approximate fixed point $x$ of $G(x)$ satisfying $\|G(x)-x\|\le\varepsilon'$ is an $\eps$ approximate solution of the variational inequality of $F(x)$. The Lipschitzness of $G$ is no more than $1+L_1$ and we can substitute $L_1\to 1+L_1$ in Corollary \ref{cor:with_boundary}. This is because for any $x$, if the algorithm returns a point $x$ with $\|G(x)\|<1$, then $G(x)=x+F(x)$ and $\|F(x)\|\le\varepsilon'$. It is indeed an $\varepsilon$-approximate solution of the variational inequality. Otherwise, if $\|G(x)\|=1$, then $x+F(x)=G(x)\cdot\lambda$ where $\lambda\in[1,L_0+1]$. Therefore, for any $y$, 
    \[
    \begin{aligned}
        &\langle F(x),y-x\rangle=\langle \lambda\cdot G(x)-x,y-x\rangle\\
        =&\langle \lambda\cdot G(x)-x,y-G(x)\rangle+\langle \lambda\cdot G(x)-x,G(x)-x\rangle\\
        =&\langle (\lambda-1)\cdot G(x),y-G(x)\rangle+\langle G(x)-x,y-G(x)\rangle+\langle \lambda\cdot G(x)-x,G(x)-x\rangle\\
        =&\langle (\lambda-1)\cdot G(x),y-G(x)\rangle+\langle (\lambda-1)\cdot G(x)+y-x,G(x)-x\rangle
    \end{aligned}
    \]
    We know that since $G(x)$ is on the boundary, $\langle (\lambda-1)\cdot G(x),y-G(x)\rangle\le 0$. Also, 
    \[
\begin{aligned}
\langle (\lambda-1)\cdot G(x)+y-x,G(x)-x\rangle
&\le\|(\lambda-1)\cdot G(x)+y-x\|\cdot\|G(x)-x\|\\
&\le(\lambda-1+2)\cdot\varepsilon'\le(L_0+2)\cdot\varepsilon'=\varepsilon.
\end{aligned}
\]
    Thus we have proved the claim. 
    
    Now we finish the rest. We know that for any surface unit $\mathrm d\sigma$ on $\mathbb S^{n-1}$, the volume within $\cB(0,1+L_0)\setminus\cB(0,1)$ that projects to it is $\mathrm d\sigma\cdot\frac{(1+L_0)^n-1}{n}\le\mathrm d\sigma\cdot\frac{(1+L_0)^n}{n}$. So we have a function $G$ such that: for any $x$, if $\|G(x)\|<1$, the density is at most $P_v$. If $\|G(x)\|=1$, the density on the $\mathbb S^{n-1}$ is at most $\frac{(1+L_0)^n}{n}P_v$. Therefore, putting $P_v\cdot (1+L_0)^n/n$ to $P_v'$ to the previous corollary yields the result.
\end{proof}

%% file: dynamics.tex
\section{Smoothed Analysis of Path-Following Methods}\label{sec:path-following}

In this section, we analyze the standard homotopy-based path-following algorithm for computing variational inequalities and fixed points in the smoothed model. The algorithm is the classical Euler-predictor/Newton-corrector continuation scheme (see, e.g.,~\cite{allgower2003introduction,garcia1981pathways}) applied to the linear homotopy between the identity and $F$~\cite{kellogg1976constructive,chow1978finding}; the only modification is a stopping rule at the boundary $\|x\|=1$, tailored to the VI setting. Our contribution is the quantitative analysis of this method under smoothed analysis, rather than the algorithm itself. We begin by introducing the continuous path through a standard homotopy (\cref{sec:continuous}), and then show how to track it by a predictor and a Newton corrector (\cref{sec:discret-alg}).
We first analyze the discrete algorithm (\cref{sec:discrete}) under a lower bound on the smallest row singular value of the homotopy Jacobian. We then prove this bound and a path-length bound directly for the homotopy under random perturbations. In the subsequent sections, we show that this condition holds with high probability under random perturbations of the function (\cref{sec:regularization-perturbation}).
Finally, we conclude by proving the main theorem statements in \cref{sec:final-analysis}.

\subsection{Differential Equation Solution and a Discretization}\label{sec:optimization}
In this section, we construct an algorithm for the approximate variational inequality problem of a vector field $F:\mathbb{R}^n\rightarrow \mathbb{R}^n$, as defined in Section \ref{subsec:problem-overview}. 
Recall that we make the following assumptions on $F$ (Assumptions \ref{asm:F-bounded}, \ref{asm:F-Lip}, \ref{asm:F-smooth}, \ref{asm:F-zero-first-order}):

\begin{itemize}
    \item Boundedness: $||F(x)||\le L_0$ for all $x\in\cB(0,1)$. 
    \item $L_1$-Lipschitzness: For all $x,y$, $\|F(x)-F(y)\|\le L_1\|x-y\|$.
    \item $L_2$-Smoothness: The Jacobian of $F$ is $L_2$-Lipschitz, for all $x,y$, $||J_F (x) - J_F (y)||_{op} \le L_2(||x - y||)$.
    \item Zeroth-order and first-order oracle to $F$.
\end{itemize}

Throughout the paper, we use the following definitions. We denote both the scalar $0$ and the zero vector by $0$, where the meaning should be clear from the context.
We denote by $\mathcal{B}_d(x,r)$ the $\ell_2$ ball in $\mathbb{R}^d$ with radius $r$ centered at $x$. If $d$ is not mentioned, by default $d=n$. Denote by $\hS^{k}$ the unit sphere of radius $1$ in $\mathbb{R}^{k+1}$.
We denote the Jacobian of a function $M:\hR^{m_1}\to\hR^{m_2}$ by $J_M$. 
For any vector $v\ne 0$, we denote the vector $\hat v$ as the unit vector with the same direction as $v$.
Throughout Section \ref{sec:path-following}, $\mathcal C>1$ denotes a universal constant whose value may change from one occurrence to another.

\subsubsection{Continuous Dynamics}\label{sec:continuous}

First, we construct continuous dynamics for the problem which defines a path to the solution. Then, we construct an iterative procedure that goes along the path.

At a high level, our strategy is as follows. We define a path $\gamma$ consisting of points $x$ that are collinear with $F(x)$, and follow the branch with $F(x)=\lambda x$, $\lambda\ge0$, until its first zero of $F$ or its first intersection with the unit sphere. The former gives a zero of $F$, and the latter gives a boundary solution to the variational inequality.

Write $a=L_0+L_1$. Because $L_0$ and $L_1$ are upper bounds, we can without loss of generality, assume $a\ge 1$. Let $\widehat F=F/a$ and define
\[
K(x,t)=(1-t)x-t\widehat F(x).
\]
For $0<t\le1$ and $x\ne0$, the equation $K(x,t)=0$ is equivalent to $F(x)=\lambda x$, where $\lambda=a\frac{1-t}{t}$.
Thus the projection onto the $x$ coordinates gives the same positive-collinearity path. At $t=0$, the unique zero is $(0,0)$; at $t=1$, we have $F(x)=0$. If $F(0)=0$, we return $0$ immediately. Otherwise, we follow the branch leaving $(0,0)$ into $t>0$. Denote its unit-speed parameterization by $\gamma:[0,T]\to\mathbb R^{n+1}$, where $T$ is the total length of the path. We stop this branch at its first point with $t=1$ or $\|x\|=1$.

The coordinate $t$ need not increase along this curve. We therefore use arclength to orient and track the path, rather than choosing increasing values of $t$. The normalization by $a$ leaves the projected path unchanged and prevents its numerical conditioning from depending on a common rescaling of $F$.

\subsubsection{Regularity of the Homotopy}\label{subsubsec:homotopy-regularity}
The homotopy Jacobian is the $n\times(n+1)$ matrix
\[
J_K(x,t)=\big[(1-t)I-tJ_F(x)/a\quad -x-F(x)/a\big].
\]
Since $L_0/a+L_1/a=1$, its operator norm is at most $L_K=3$ on $\cB(0,1)\times[0,1]$; thus $K$ is $L_K$-Lipschitz on this cylinder. We also use the deterministic bound
\[
L_J=2L_2/a+2(1+L_1/a)
\]
for the Lipschitz constant of $J_K$ on $\cB(0,1)\times[-1,2]$. Indeed, for any two points $(x,t),(y,s)\in\cB(0,1)\times[-1,2]$, the differences between of $J_K(x,t)$ and $J_K(y,s)$ have norms at most
$2(L_2/a)\|x-y\|+(1+L_1/a)|t-s|$ for the left square and $(1+L_1/a)\|x-y\|$ and for the last column respectively. 

Unlike a square Jacobian, $J_K$ can have a one-dimensional kernel while all of its row singular values are positive. This is the regularity condition needed by the algorithm.

\begin{property}\label{asm:homotopy-gap}
At every zero $(x,t)$ of $K$ in $\cB(0,1)\times[0,1]$, the smallest row singular value satisfies $\sigma_n(J_K(x,t))\ge\theta>0$.
\end{property}

This lower bound for the singular value is to make the path have a good condition, including no self-intersections, bounded curvature, and smoothness. It also guarantees success in proceeding with our algorithm. Now we give a bound for $\theta$ for our settings: 

\begin{lemma}
Under Assumptions \ref{asm:F-bounded}, \ref{asm:F-Lip}, \ref{asm:F-smooth}, \ref{asm:bounded-density}, and \ref{asm:bounded-jacobian-density}, there is a universal constant $\mathcal C>1$ such that, for $n\ge2$ and $0<p<1$, the deterministic choice
\[
\theta=\mathcal C^{-n}\left(\frac{a^{2n}P_d}{n^n}\right)^{-1}\frac{p}{L_J},
\qquad a=L_0+L_1,
\]
lies in $(0,1/8]$ and satisfies Property \ref{asm:homotopy-gap} with probability at least $1-p$.
\end{lemma}
This lemma is proved in Lemma \ref{lem:delta}. 

\begin{lemma}\label{asm:full rank}
Under Property \ref{asm:homotopy-gap}, the outgoing zero branch of $K$ is a regular curve. If $F(0)\ne0$, its projection is a regular positive-collinearity curve, including at $0$. If the followed branch has finite length, its terminal point satisfies $t=1$ or $\|x\|=1$, and gives a variational inequality solution.
\end{lemma}
\begin{proof}
The row singular-value bound and the implicit function theorem give a local regular curve at every zero. At $(0,0)$, the kernel is spanned by $(F(0)/a,1)$, so it has a unique outgoing direction into $t>0$. At a nonzero zero, $\partial_t K=-x/t$; thus a nonzero tangent cannot have zero $x$ component. The projected curve is regular, and the value of $t$ is uniquely determined by a nonzero positive-collinearity point, so projection does not introduce a self-intersection.

A finite-length maximal branch has a limit in the compact cylinder. At an interior limit, the implicit function theorem extends it, contradicting maximality. The slice $t=0$ contains only the initial zero, which cannot be revisited by the regular oriented branch. Hence its terminal point has $t=1$ or $\|x\|=1$. These are respectively a zero of $F$ and a boundary point with $F(x)=\lambda x$, $\lambda\ge0$.
\end{proof}

\subsubsection{Path Following Algorithm}\label{sec:discret-alg}

We now explain how an iterative algorithm can follow the lifted curve $\gamma$. We can evaluate $K$ and $J_K$ using the zeroth-order and first-order oracles of $F$. At an approximate curve point $x_i=(\bar x_i,t_i)$, we compute a unit vector $v_i\in\ker J_K(x_i)$, with its sign chosen consistently with the previous direction.

First, we \textbf{go forward} to the next step $y_i=x_i+hv_i$. $h$ is the step size that we are going to define later. Second, we \textbf{push back} by one Newton correction. We use Newton's method to correct the point. We solve the linear system for $J_K(y_i)c_i=-K(y_i)$ with $v_i\perp c_i$ (and this $v_i$ is unique up to direction because of $\theta$). Then we update $x_{i+1}=y_i+c_i$. We refer the reader to Section \ref{sec:discrete}, which proves two points: for all ``going forward'' processes, the path does not deviate too much, and ``push back'' does not reverse progress too much. So overall the progress on the path is lower bounded.

The path begins at zero, so we can initially get the starting point and the tangent. Since $J_K(0,0)=[I~~-F(0)/a]$, the kernel is in the direction $(F(0)/a,1)$. So, we should set $x_0=(0,0)$ and $v_0=v_{-1}=\frac{1}{\sqrt{1+(\|F(0)\|^2/a^2)}}\cdot(F(0)/a,1)$.

Next, we need to determine when to stop the algorithm based on the value of $F(x)$. We formalize this idea with the following stopping rule, called the ``Predicate function''. The predicate function queries the value of $F$ and outputs one of three possible states: ``Run'', ``Stop'', or ``Project''. A ``Stop'' point is returned directly, and a ``Project'' point is projected radially onto the unit sphere. In the following lemma, we show such a predicate exists by querying the zeroth- and first-order oracle of the function that outputs an $\varepsilon$-approximate variational inequality solution.

\begin{lemma}[Stopping rule predicate]\label{asm: oracle}
Let $a=L_0+L_1>0$ and $\varepsilon>0$.
There is a stopping predicate
$H:\cB(0,1)\to\{\mathrm{Run},\mathrm{Stop},\mathrm{Project}\}$,
computable with at most two zeroth-order oracle queries.
If $H(x)=\mathrm{Stop}$, we return $x$; if
$H(x)=\mathrm{Project}$, then $x\ne0$ and we return $x/\|x\|$.
Both types of output are $\varepsilon$-approximate variational
inequality solutions.
\end{lemma}
\begin{proof}
For $x\in\cB(0,1)$, its variational inequality gap is
\[
g(x):=\max_{y\in\cB(0,1)}\langle F(x),y-x\rangle
=\|F(x)\|-\langle F(x),x\rangle.
\]
We define $H$ as follows:
\begin{enumerate}
    \item If $\|F(x)\|\le\varepsilon/2$, then
    $H(x)=\mathrm{Stop}$.
    \item Otherwise, if $x\ne0$, $\|x\|\ge1-\varepsilon/(2a)$,
    and $g(x/\|x\|)\le\varepsilon$, then
    $H(x)=\mathrm{Project}$.
    \item In all remaining cases, $H(x)=\mathrm{Run}$.
\end{enumerate}
If $H(x)=\mathrm{Stop}$, then
$g(x)\le2\|F(x)\|\le\varepsilon$.
If $H(x)=\mathrm{Project}$, the returned point $x/\|x\|$
has gap at most $\varepsilon$ by the second test.
The first test uses $F(x)$, and the second uses at most one
additional query at $x/\|x\|$; both query points are in the unit ball.
\end{proof}

We now give the algorithm. We remind the reader that every step starts with $x_i$,  the next direction is $v_i$, the intermediate step $y_i$,
and corrections $c_i$ are vectors in $\mathbb R^{n+1}$.
We write $x_i=(\bar x_i,t_i)$, where $\bar x_i\in\mathbb R^n$
is the spatial component. Thus $H$ is evaluated at $\bar x_i$,
whereas $K$ and $J_K$ are evaluated at  $x_i$.
The step size uses the two bounds already available:
\[
h=\min\left\{\frac{\theta}{64L_J},\frac{\varepsilon}{128a}\right\}.
\]
The first bound controls the change in the tangent; the second
ensures that a step cannot skip the region where the predicate accepts.
These claims are proved in Section \ref{sec:discrete}.

\begin{algorithm}[H]
\caption{Follow-The-Path}
\label{alg:discrete}
\begin{algorithmic}[1]
\Require Function $F$, bounds $L_0,L_1,L_2$, accuracy $\varepsilon>0$,
row-gap bound $0<\theta\le1$, and the stopping predicate $H$ of Lemma \ref{asm: oracle}
\If{$\|F(0)\|\le\varepsilon/2$}
\State \Return $0$
\EndIf
\State $a\gets L_0+L_1$, $K(x,t)\gets(1-t)x-tF(x)/a$
\State $L_J\gets2L_2/a+2(1+L_1/a)$,
$h\gets\min\{\theta/(64L_J),\varepsilon/(128a)\}$
\State $x_0\gets(0,0)$,
$v_{-1}\gets(F(0)/a,1)/\sqrt{1+\|F(0)\|^2/a^2}$
\State $i\gets0$
\While{$H(\bar x_i)=\mathrm{Run}$}
\State Choose a unit $v_i\in\ker J_K(x_i)$ with $\langle v_i,v_{i-1}\rangle>0$
\State $y_i\gets x_i+hv_i$ \Comment{Predictor step}
\State Solve $J_K(y_i)c_i=-K(y_i)$ with $v_i^\top c_i=0$.\Comment{Corrector step}
\State $x_{i+1}\gets y_i+c_i$
\State $i\gets i+1$
\EndWhile
\If{$H(\bar x_i)=\mathrm{Project}$}
\State \Return $\bar x_i/\|\bar x_i\|$
\Else
\State \Return $\bar x_i$
\EndIf
\end{algorithmic}
\end{algorithm}

Now we begin our main theorem.

\begin{theorem}\label{thm:optimization}
Under Assumptions \ref{asm:F-bounded}, \ref{asm:F-Lip},
\ref{asm:F-smooth}, and \ref{asm:F-zero-first-order}, suppose the
followed lifted branch has finite length $T$ and
$\sigma_n(J_K)\ge\theta$ along this branch, where $0<\theta\le1$.
With the stopping predicate of Lemma \ref{asm: oracle},
Algorithm \ref{alg:discrete} finds an $\varepsilon$-approximate
variational inequality solution using
\[
O\!\left(1+T\max\left\{\frac{L_J}{\theta},\frac{a}{\varepsilon}\right\}\right)
\]
zeroth-order and first-order oracle calls.
Each iteration uses a constant number of these calls and
polynomially many arithmetic operations in $n$ for a kernel computation
and a linear-system solve. The statement uses exact real arithmetic;
all queried points lie in the unit ball.
\end{theorem}

The following section \ref{sec:discrete} proves this theorem as if we take $\theta$ as granted. In
Section \ref{sec:regularization-perturbation}, we bound $\theta$ and $T$
directly for $K$ (Lemmas \ref{lem:delta} and \ref{lem:length}).

%% file: discrete.tex
\subsection{Algorithm Analysis}\label{sec:discrete}

In this section, we prove Theorem \ref{thm:optimization} if we are given $\theta$ and $T$.
We first establish the geometric facts needed for tracking the path
and for the path-length bound in Section \ref{sec:regularization-perturbation}.
We then analyze the forward and Newton steps and prove that the
algorithm reaches a point accepted by the stopping predicate.
We also remind the reader that $\gamma:[0,T]\to\mathbb R^{n+1}$ is the curve and we may abuse the notation of $\gamma$ either as a function or as a geometric object.

\subsubsection{Correctness Analysis}\label{sec:correctness}

We begin with a property for normal disks around the zero set.
The property also applies to points on different curve components,
which will be useful in the path-length analysis. 

\begin{lemma}\label{lem:unique proj}
Let $\Gamma$ be a collection of regular zero-curve portions of $K$
in $\cB(0,1)\times[0,1]$, with $\sigma_n(J_K(p))\ge\theta$
at every $p\in\Gamma$. For any $q\in\mathbb R^{n+1}$, there is
at most one $p\in\Gamma$ such that $\|q-p\|<\theta/L_J$
and $q-p$ is perpendicular to the tangent at $p$.
\end{lemma}
\begin{proof}
Suppose $p_0\ne p_1$ are two such points. Let
$d=p_1-p_0$ and $r=\|d\|$. Since $K(p_0)=K(p_1)=0$,
By Taylor's formula, we have
\[K(p_1)=K(p_0)+J_K(p_0)d+\int_{0}^1(J_K(p_0+td)-J_K(p_0))d\mathrm dt.\]
Since we have $K(p_0)=K(p_1)=0$, we have
\[\begin{aligned}
    &\|J_K(p_0)d\|=\left\|\int_{0}^1(J_K(p_0+td)-J_K(p_0))d\mathrm dt\right\|\le\int_{0}^1\left\|(J_K(p_0+td)-J_K(p_0))\right\|\cdot \|d\|\mathrm dt\\
    \le&\int_{0}^1L_Jt\|d\|\mathrm dt\cdot\|d\|=L_Jr^2/2.
\end{aligned}\]
Let $d_{\perp}$ and $d_{\parallel}$ be the perpendicular and parallel components to the tangent at $p_0$. Therefore $J_K(p_0)d_{\parallel}=0$ and by the assumption, $\|J_K(p_0)d_{\perp}\|\ge \theta\|d_{\perp}\|$. Therefore, we have $\|d_\perp\|\le L_Jr^2/(2\theta)$.
Because $q-p_0$ is also perpendicular to that tangent (so perpendicular to $d_{\parallel}$), we have
\[
|\langle q-p_0,d\rangle|
=|\langle q-p_0,d_\perp\rangle|\le \|q-p_0\|\cdot\|d_{\perp}\|<r^2/2.
\]
By symmetry, we can deduce that $|\langle q-p_1,d\rangle|
<r^2/2$. However, the difference between these two inner products is
$\langle p_1-p_0,d\rangle=r^2$, a contradiction.
\end{proof}

From the above lemma, we can also deduce that the curvature of the path is at most $L_J/\theta$. Denote the upper bound of the curvature of $\gamma$ to be $C_\gamma = L_J/\theta$ (so the radius of curvature is at least $1/C_{\gamma}$). The next lemma controls how the tangent changes near the curve.

\begin{lemma}\label{lem:convex}
If $x\in\cB(0,1)\times[-1,2]$ and for some $s\in[0, T]$, 
$\|x-\gamma(s)\|\le\theta/(2L_J)$, then $J_K(x)$ has full row rank.
Its unit kernel vector $v$, with sign chosen so that
$\langle v,\gamma'(s)\rangle\ge0$, satisfies
\[
\|v-\gamma'(s)\|\le\sqrt2C_\gamma\|x-\gamma(s)\|.
\]
\end{lemma}
\begin{proof}
The smallest row singular value at $x$ is at least
$\theta-L_J\|x-\gamma(s)\|\ge\theta/2$. So $J_K(x)$ has full row rank. Therefore, the $v$ taken is uniquely defined.

Now we consider the unit kernel $v$ decomposing to $v=b\gamma'(s)+u$, where $u\perp\gamma'(s)$ and $b\ge 0$.
Since $J_K(x)v=0$, the row gap on the curve gives that $\theta\|u\|\le\|J_K(\gamma(s))v\|
\le L_J\|x-\gamma(s)\|$, and by the fact that $b\ge 0$, using trigonometry we can calculate that $\|v-\gamma'(s)\|\le\sqrt2\|u\|$. Therefore, we have
\[\|v-\gamma'(s)\|\le\sqrt2\|u\|\le \sqrt{2}\|J_K(\gamma(s))v\|/\theta\le \sqrt{2}L_J/\theta\cdot\|x-\gamma(s)\|=\sqrt{2}C_{\gamma}\|x-\gamma(s)\|.\]
\end{proof}

Finally, we give a bound of the inverse of the square matrix obtained by adding
a row. This is used for the Newton step.
\begin{lemma}\label{lem:newton-inverse}
For any $s\in[0,T]$, For any unit vector $v$ with $\langle v,\gamma'(s)\rangle\ge3/4$,
we have
\[
\left\|\begin{pmatrix}J_K(\gamma(s))\\v^\top\end{pmatrix}^{-1}\right\|
\le4/\theta.
\]
\end{lemma}
\begin{proof}
Again, we decompose any vector $d=b\gamma'(s)+u$
with $u\perp\gamma'(s)$, and write
$r=J_K(\gamma(s))d$ and $\beta=\langle v,d\rangle$.
Then by the property of $\theta$ we have $\|u\|\le\|r\|/\theta$, and thus we have
\[
\frac34|b|\le|b\langle v,\gamma'(s)\rangle|=|\beta-\langle v,u\rangle|\le|\beta|+\|u\|\le|\beta|+\|r\|/\theta.
\]
Therefore we know that $|b|\le \frac{4}{3}(|\beta|+\|r\|/\theta)$.
Since $\theta\le1$, these inequalities imply that
\[\|d\|\le \|u\|+|b|\le \frac{4}{3}|\beta|+\frac{7}{3}\|r\|/\theta\le\frac{11}{3}\|(r,\beta)\|/\theta\le 4\|(r,\beta)\|/\theta,\]
and this concludes the proof.
\end{proof}

\subsubsection{Runtime Analysis}\label{sec:runtime}

We remind the reader that $x_i\in\mathbb R^{n+1}$ is the actual iterate of
Algorithm \ref{alg:discrete}. Its ``going forward'' step is
$y_i=x_i+hv_i$, where $v_i$ is the unit kernel direction chosen by
the algorithm. Its ``pushing back'' step is $x_{i+1}=y_i+c_i$,
where $v_i\perp c_i$ and $J_K(y_i)c_i=-K(y_i)$. 

We define the points on the curve $\gamma(s_i)$ to compare with $x_i$. Set $s_0=0$. For each  $x_i$, $i\ge1$, let $s_i$ be the unique parameter (the uniqueness is proved in Lemma \ref{lem:go-forward}) $s_i\in(s_{i-1}+h/2,\ s_{i-1}+2h)$ satisfying $\langle v_{i-1},\gamma(s_i)-x_i\rangle=0$. That is, if we want $x_i'=y_{i-1}+c'_{i-1}$ where $x'_i$ is on the curve and $c'_{i-1}$ can be any small vector that is perpendicular to $v_{i-1}$, we will get $x_i'=\gamma(s_i)$. However, the Newton step sends it to $x_i$. So this difference measures where $y_{i-1}$ should be sent to versus where it is actually sent to.

In the proofs below, we denote
\[
E:=64C_\gamma^3h^4,
\qquad
h=\min\left\{\frac{\theta}{64L_J},\frac{\varepsilon}{128a}\right\}.
\]

\paragraph{Go Forward.}
We first locate the intersection with the hyperplane through
$y_i$ perpendicular to $v_i$, and bound the distance from $y_i$
to this curve point.
\begin{lemma}\label{lem:go-forward}
Suppose $x_i\in\cB(0,1)\times[-1,2]$,
$\|x_i-\gamma(s_i)\|\le E$, and $s_i+2h\le T$.
Let $v_i$ be the unit kernel vector of $J_K(x_i)$ oriented toward
$\gamma'(s_i)$, and let $y_i=x_i+hv_i$.
There is a unique $s_{i+1}\in(s_i+h/2,s_i+2h)$ such that
\[
\langle v_i,\gamma(s_{i+1})-y_i\rangle=0.
\]
This curve point satisfies
\[
\|y_i-\gamma(s_{i+1})\|\le2C_\gamma h^2.
\]
\end{lemma}
\begin{proof}
Since $E\le h\le1/(64C_\gamma)<1/(2C_\gamma)$, Lemma \ref{lem:convex} applies and gives
\[
\|v_i-\gamma'(s_i)\|
\le\sqrt2C_\gamma\|x_i-\gamma(s_i)\|
\le\sqrt2C_\gamma E.
\]
For $0\le l\le2h$, the curvature bound gives
\[
\|\gamma'(s_i+l)-\gamma'(s_i)\|
\le\int_0^l\|\gamma''(s_i+r)\|\,dr
\le C_\gamma l
\le2C_\gamma h.
\]
Therefore, using $\|v_i\|=1$ and the triangle inequality, we obtain
\[
\begin{aligned}
&\langle v_i,\gamma'(s_i+l)\rangle=1+\langle v_i,\gamma'(s_i+l)-v_i\rangle\ge1-\|\gamma'(s_i+l)-v_i\|\\
\ge&1-\|\gamma'(s_i+l)-\gamma'(s_i)\|
       -\|\gamma'(s_i)-v_i\|\ge1-\sqrt2C_\gamma E-2C_\gamma h.
\end{aligned}
\]
Finally, since $E\le h$ and $C_\gamma h\le1/64$,
\[
\sqrt2C_\gamma E+2C_\gamma h
\le(\sqrt2+2)C_\gamma h
\le\frac{\sqrt2+2}{64}<\frac14.
\]
Hence $\langle v_i,\gamma'(s_i+l)\rangle\ge3/4$ for all $0\le l\le 2h$. Now we consider the scalar function
$f(l)=\langle v_i,\gamma(s_i+l)-x_i\rangle$. By the assumption $\|x_i-\gamma(s_i)\|\le E$, we have
$|f(0)|
=|\langle v_i,\gamma(s_i)-x_i\rangle|
\le\|v_i\|\,\|\gamma(s_i)-x_i\|
\le E$.
Thus $-E\le f(0)\le E$. Using $3/4\le f'(\ell)\le1$, we obtain
\[
\begin{aligned}
f(h/2)
&=f(0)+\int_0^{h/2}f'(\ell)\,d\ell
\le E+\frac h2,\\
f(2h)
&=f(0)+\int_0^{2h}f'(\ell)\,d\ell
\ge-E+\frac34\cdot2h
=\frac {3h}2-E.
\end{aligned}
\]
Since $E\le h/4096<h/2$, it follows that
$E+\frac h2<h<\frac {3h}2-E$. Therefore, there exists a unique $d$ such that $h/2<d<2h$ and $f(d)=h$, and we can define $s_{i+1}=s_i+d$.

Now we bound the distance. The curvature bound gives
$\|\gamma(s_i+h)-\gamma(s_i)-h\gamma'(s_i)\|\le C_\gamma h^2/2$.
Also we have $\|x_i-\gamma(s_i)\|\le E$, by Lemma \ref{lem:convex} we have $\|v_i-\gamma'(s_i)\|\le \sqrt{2}C_{\gamma}E$. Combining all of these and the bound for $E$ and $h$, we obtain
\[
\begin{aligned}
&\|y_i-\gamma(s_i+h)\|= \|x_i+hv_i-\gamma(s_i+h)\|\\
\le&\|x_i-\gamma(s_i)\|+h\|v_i-\gamma'(s_i)\|+\|\gamma(s_i)+h\gamma'(s_i)-\gamma(s_i+h)\|\\
\le &E+\sqrt2C_\gamma hE+C_\gamma h^2/2
\le \frac{3}{4} C_\gamma h^2.
\end{aligned}
\]

Since we have reparametrized $\gamma$ as a unit-speed curve, so
\[
\begin{aligned}
&\|y_i-\gamma(s_{i+1})\|\le\|y_i-\gamma(s_{i}+h)\|+\|\gamma(s_i+h)-\gamma(s_{i+1})\|\\
\le&\|y_i-\gamma(s_i+h)\|+|s_{i+1}-s_i-h|
\le \frac{3}{4}C_\gamma h^2+|s_{i+1}-s_i-h|.   
\end{aligned}
\]

So it suffices to prove $|s_{i+1}-s_i-h|\le C_\gamma h^2$. This is by the Lipschitzness: since we have 
\[f(h)-f(s_{i+1}-s_i)=f(h)-h=\langle v_i, \gamma(s_i+h)-x_i-hv_i\rangle=\langle v_i, \gamma(s_i+h)-y_i\rangle\]

and $v_i$ is the unit vector, so we have \[|f(h)-f(s_{i+1}-s_i)|=|\langle v_i, \gamma(s_i+h)-y_i\rangle|\le \|\gamma(s_i+h)-y_i\|\le \frac{3}{4}C_{\gamma}h^2.\]

Thus, we have, because $f'$ is lower bounded by $3/4$, we have $|s_{i+1}-s_i-h|\le \frac{4}{3}|f(h)-f(s_{i+1}-s_i)|\le \frac{4}{3}\cdot \frac{3}{4}C_{\gamma}h^2=C_\gamma h^2$. 
\end{proof}

\paragraph{Push Back.}
By the construction, the point $\gamma(s_{i+1})$ and $x_{i+1}$ both lie in the same hyperplane through $y_i$ and perpendicular to $v_i$.
We use this fact to prove that $s_{i+1}$ satisfies the condition for $x_{i+1}$ in Lemma \ref{lem:go-forward}, that is, $\|\gamma(s_{i+1})-x_{i+1}\|\le E$.
\begin{lemma}\label{lem:push-back}
Under the assumptions of Lemma \ref{lem:go-forward}, suppose also
that $y_i\in\cB(0,1)\times[-1,2]$.
The Newton step $c_i$ in Algorithm \ref{alg:discrete} is
well defined, and $x_{i+1}=y_i+c_i$ satisfies
\[
\|x_{i+1}-\gamma(s_{i+1})\|
\le16C_\gamma^3h^4\le E,
\qquad
\|c_i\|\le3C_\gamma h^2.
\]
\end{lemma}
\begin{proof}
Write $p=\gamma(s_{i+1})$ and let
$M_i=\begin{pmatrix}J_K(y_i)\\v_i^\top\end{pmatrix}$
be the matrix for the Newton step.
By the proof of Lemma \ref{lem:go-forward}, 
$\langle v_i,\gamma'(s_{i+1})\rangle\ge3/4$.
By Lemma \ref{lem:newton-inverse}, $\begin{pmatrix}J_K(p)\\v_i^\top\end{pmatrix}^{-1}$
has norm at most $4/\theta$.
By Assumption \ref{asm:F-smooth}, we know that $\|J_K(p)-J_K(y_i)\|\le L_J\|y_i-p\|$. Also, by Lemma \ref{lem:go-forward}, $\|y_i-p\|\le 2C_\gamma h^2$. Hence
\[
\left\|\begin{pmatrix}J_K(p)\\v_i^\top\end{pmatrix}^{-1}
\left(M_i-\begin{pmatrix}J_K(p)\\v_i^\top\end{pmatrix}\right)\right\|\le 4/\theta\cdot L_J\|y_i-p\|
\le8C_\gamma^2h^2<\frac12.
\]
Let $A=\begin{pmatrix}J_K(p)\\v_i^\top\end{pmatrix}$ and $B=\begin{pmatrix}J_K(p)\\v_i^\top\end{pmatrix}^{-1}
\left(M_i-\begin{pmatrix}J_K(p)\\v_i^\top\end{pmatrix}\right)$. Since $M_i=A(I+B)$ and $\|B\|<1/2$, we have 
\[\|M_i^{-1}\|\le\|A^{-1}\| \cdot \|(I+B)^{-1}\|\le\|A^{-1}\|\cdot \frac{1}{1-\|B\|}\le 8/\theta. \]

We now bound the error after this Newton step explicitly.
The fundamental theorem of calculus gives
\[
\begin{aligned}
K(p)&=K(y_i)+J_K(y_i)(p-y_i)+\int_0^1
\bigl[J_K(y_i+t(p-y_i))-J_K(y_i)\bigr](p-y_i)\mathrm dt.
\end{aligned}
\]
Let the remainder vector be $R_i:=\int_0^1
\bigl[J_K(y_i+t(p-y_i))-J_K(y_i)\bigr](p-y_i)\mathrm dt$. Again, Jacobian is $L_J$-Lipschitz along this segment, so
\[
\begin{aligned}
&\left\|\int_0^1
\bigl[J_K(y_i+t(p-y_i))-J_K(y_i)\bigr](p-y_i)\mathrm dt\right\|\\
\le&\int_0^1
\left\|\bigl[J_K(y_i+t(p-y_i))-J_K(y_i)\bigr]\right\|\|p-y_i\|\mathrm dt\\
\le &\int_0^1 L_Jt\|p-y_i\|^2\mathrm dt=\frac{L_J}{2}\|p-y_i\|^2.
\end{aligned}
\]

Here we need to note that the segment connecting $p$ and $y_i$ is totally in the region and hence the $J_K(\cdot)$ is defined. 

Since $K(p)=0$, $J_K(y_i)c_i=-K(y_i)$, and
$\langle v_i,p-y_i\rangle=\langle v_i,c_i\rangle=0$, we have
\[
M_i(x_{i+1}-p)=\begin{pmatrix}R_i\\0\end{pmatrix}.
\]
Therefore,
\[
\|x_{i+1}-p\|
\le\frac8\theta\frac{L_J}{2}\|p-y_i\|^2
\le16C_\gamma^3h^4\le E.
\]
Finally,
$\|c_i\|=\|x_{i+1}-y_i\|\le\|y_i-p\|+\|x_{i+1}-p\|
\le2C_\gamma h^2+16C_\gamma^3h^4\le3C_\gamma h^2$. This concludes the proof.
\end{proof}

\paragraph{Combining the two steps.}
We now prove that the recursively defined parameters $s_i$ exist
throughout the algorithm, that all oracle queries stay in the
allowed domain, and that progress along the curve leads to termination. Since we are dealing with a predicate oracle, we are going to remind in Lemma \ref{asm: oracle}, our oracle $H$ in the proof is: ($g(x)=\max_{y\in\cB(0,1)}\langle F(x),y-x\rangle=\|F(x)\|-\langle F(x),x\rangle$.) 
\begin{enumerate}
    \item If $\|F(x)\|\le\varepsilon/2$, then
    $H(x)=\mathrm{Stop}$.
    \item Otherwise, if $x\ne0$, $\|x\|\ge1-\varepsilon/(2a)$,
    and $g(x/\|x\|)\le\varepsilon$, then
    $H(x)=\mathrm{Project}$.
    \item In all remaining cases, $H(x)=\mathrm{Run}$.
\end{enumerate}
Notice that for any point $x=(\bar x,t)$, the oracle only tests against the spatial component, which is $\bar x$. 
\begin{lemma}\label{lem:combine}
For every $x_i$ reached by Algorithm \ref{alg:discrete},
the recursively defined parameter $s_i\in[0,T]$ exists uniquely and
satisfies $\|x_i-\gamma(s_i)\|\le E$.
Whenever $H(\bar x_i)=\mathrm{Run}$, the next forward and Newton
steps are well defined, the chosen direction satisfies
$\langle v_i,\gamma'(s_i)\rangle>0$, and
$h/2<s_{i+1}-s_i<2h$.
All oracle queries stay in the unit ball, and the algorithm
terminates after at most $\lceil2T/h\rceil$ predictor--corrector cycles.
\end{lemma}

\begin{proof}
We prove the assertions about the iterates and directions by
induction, and then use the progress bound to prove termination.
Initially, $x_0=\gamma(0)$, $s_0=0$, and
$v_{-1}=\gamma'(0)$, so the initial $\|x_0-\gamma(s_0)\|$ is zero.
Suppose that the parameters and iterates have been defined through
$x_i$, that
\[
x_j\in\cB(0,1)\times[-1,2],
\qquad
\|x_j-\gamma(s_j)\|\le E
\qquad (0\le j\le i),
\]
and that all preceding cycles satisfy the stated direction and
progress bounds.
If $H(\bar x_i)\ne\mathrm{Run}$, the algorithm returns a correct
output by Lemma \ref{asm: oracle}.
We therefore assume $H(\bar x_i)=\mathrm{Run}$ and prove that the
next cycle preserves all these assertions.
Let $p_i$ be the spatial component of $\gamma(s_i)$.
Then $\|\bar x_i-p_i\|\le E$.

We first prove that $\|p_i\|<1-\varepsilon/(8a)$; otherwise, it is going to return Project. We prove this to show that this range will give
enough distance from the boundary to execute the next cycle and we are not going outside of the region in the next iteration.
Because the Stop test fails,
\[
\varepsilon<2\|F(\bar x_i)\|\le2L_0\le2a.
\]
Suppose, to the contrary, that
$\|p_i\|\ge1-\varepsilon/(8a)$.
Then $p_i\ne0$, and the homotopy equation gives
$F(p_i)=\lambda p_i$ for some $\lambda\ge0$.
Since $E\le\varepsilon/(32a)$,
\[
\|\bar x_i\|
\ge\|p_i\|-E
\ge1-\frac{5\varepsilon}{32a}
>1-\frac{\varepsilon}{2a}>0.
\]
Thus $\bar x_i$ satisfies the radius condition for Project in
Lemma \ref{asm: oracle}.

To show that $\bar x_i$ also satisfies the gap condition for Project,
write $q=\bar x_i/\|\bar x_i\|$.
The triangle inequality gives
\[
\begin{aligned}
\|q-p_i\|
&\le\|q-\bar x_i\|+\|\bar x_i-p_i\|=1-\|\bar x_i\|+\|\bar x_i-p_i\|\\
&\le1-\|p_i\|+2\|\bar x_i-p_i\|\le\frac{\varepsilon}{8a}+2E
\le\frac{3\varepsilon}{16a}.
\end{aligned}
\]
Recall that $g(u)=\|F(u)\|-\langle F(u),u\rangle$.
For $u,v\in\cB(0,1)$, boundedness and Lipschitzness of $F$ imply
\[
\begin{aligned}
|g(u)-g(v)|
&\le(1+\|u\|)\|F(u)-F(v)\|
  +\|F(v)\|\,\|u-v\|\\
&\le(L_0+2L_1)\|u-v\|.
\end{aligned}
\]
Since $F(p_i)$ and $p_i$ are collinear, this also gives
\[
g(p_i)=\|F(p_i)\|(1-\|p_i\|)
\le\frac{L_0\varepsilon}{8a}.
\]
Therefore,
\[
\begin{aligned}
g(q)
\le g(p_i)+(L_0+2L_1)\|q-p_i\|
\le\frac{L_0\varepsilon}{8a}
 +(L_0+2L_1)\frac{3\varepsilon}{16a}
\le\frac{3\varepsilon}{8}<\varepsilon.
\end{aligned}
\]
The predicate would therefore return Project, contradicting
$H(\bar x_i)=\mathrm{Run}$.

Now, we have proved $\|p_i\|<1-\frac{\varepsilon}{8a}$ and thus $
\|\bar x_i\|\le\|p_i\|+E<1-\frac{3\varepsilon}{32a}$. We next prove that $s_i+2h<T$, so that the curve portion required
by Lemma \ref{lem:go-forward} is available.

Let $p$ be the spatial component of $\gamma(T)$.
By Lemma \ref{asm:full rank}, either $\|p\|=1$ or $F(p)=0$.
Since $\gamma$ is unit speed, $\|p-p_i\|\le T-s_i$.
If $T-s_i\le\varepsilon/(8a)$ and $\|p\|=1$, then
\[
\|p_i\|\ge\|p\|-\|p-p_i\|
\ge1-(T-s_i)\ge1-\frac{\varepsilon}{8a},
\]
contradicting the bound on $\|p_i\|$.
If instead $F(p)=0$, then
\[
\begin{aligned}
\|F(\bar x_i)\|\le L_1\|\bar x_i-p\|\le L_1\bigl(\|\bar x_i-p_i\|+\|p_i-p\|\bigr)\le L_1\left(E+\frac{\varepsilon}{8a}\right)\le\frac{5\varepsilon}{32}<\frac{\varepsilon}{2}.
\end{aligned}
\]
The predicate would return Stop, again contradicting Run.
Hence $T-s_i>\frac{\varepsilon}{8a}>2h$.

We now prove that the direction selected by the algorithm is
oriented toward $\gamma'(s_i)$, as required by
Lemma \ref{lem:go-forward}.
Since $E\le h\le\theta/(64L_J)$,
Lemma \ref{lem:convex} applies at $x_i$.
In particular, $J_K(x_i)$ has full row rank.
Let $v$ be its unit kernel vector oriented toward $\gamma'(s_i)$.
For $i\ge1$, Lemma \ref{lem:convex}, the curvature bound, and the
induction hypothesis give
\[
\begin{aligned}
\|v-v_{i-1}\|
&\le\|v-\gamma'(s_i)\|
 +\|\gamma'(s_i)-\gamma'(s_{i-1})\|
 +\|\gamma'(s_{i-1})-v_{i-1}\|\\
&\le\sqrt2C_\gamma E
 +C_\gamma(s_i-s_{i-1})
 +\sqrt2C_\gamma E\\
&\le2\sqrt2C_\gamma E+2C_\gamma h\le\frac{2\sqrt2+2}{64}<1/10.
\end{aligned}
\]
Because $v$ and $v_{i-1}$ are unit vectors, the norm above shows that the angle between $v$ and $v_{i-1}$ is at most $10$ degrees. Thus the sign rule in Algorithm \ref{alg:discrete} selects $v_i=v$.
For $i=0$, this follows directly from
$x_0=\gamma(0)$ and $v_{-1}=\gamma'(0)$.

We verify that the predictor $y_i=x_i+hv_i$ lies in the region
required by Lemma \ref{lem:push-back}. Let $x_i=(\bar x_i, t_i)$.
So $y_i$'s spatial component has norm at most
\[
\|\bar x_i\|+h
<1-\frac{3\varepsilon}{32a}
  +\frac{\varepsilon}{128a}<1.
\]
The $t$-coordinate of $x_i$ is within $E$ of the corresponding
coordinate of $\gamma(s_i)$, which belongs to $[0,1]$.
The $t$-coordinate of $y_i$ therefore belongs to
$[-E-h,1+E+h]$.
Since $E\le h$ and $h\le1/128$, this interval is contained in
$(-1,2)$.
Hence $y_i\in\cB(0,1)\times[-1,2]$, so all evaluations needed for
the forward and Newton steps are allowed.

After checking all the preconditions, we can now apply the preceding two lemmas to construct the next
$x_{i+1}$ and $s_{i+1}$.
Lemma \ref{lem:go-forward} gives a unique
$s_{i+1}\in(s_i+h/2,s_i+2h)$ satisfying
$\langle v_i,\gamma(s_{i+1})-y_i\rangle=0$.
Lemma \ref{lem:push-back} gives a unique Newton correction $c_i$
and proves
\[
\|x_{i+1}-\gamma(s_{i+1})\|\le E,
\qquad
\|c_i\|\le3C_\gamma h^2.
\]
The Newton constraint $\langle v_i,c_i\rangle=0$ implies, for every $s$,
\[
\langle v_i,\gamma(s)-x_{i+1}\rangle
=\langle v_i,\gamma(s)-y_i\rangle.
\]
Thus the unique parameter just obtained is exactly the one prescribed
by the recursive definition using $x_{i+1}$.

We finally verify that $x_{i+1}$ remains in the allowed
region, completing the induction.
Using $C_\gamma h\le1/64$, we obtain
\[
\|x_{i+1}-x_i\|\le h+\|c_i\|\le h+3C_\gamma h^2\le h(1+3/64)
<2h\le\frac{\varepsilon}{64a}.
\]
Its spatial component therefore satisfies
\[
\|\bar x_{i+1}\|
\le\|\bar x_i\|+\|x_{i+1}-x_i\|
<1-\frac{3\varepsilon}{32a}
  +\frac{\varepsilon}{64a}<1.
\]
Its $t$-coordinate is within $E$ of $[0,1]$ because
$\|x_{i+1}-\gamma(s_{i+1})\|\le E$.
Therefore $x_{i+1}\in\cB(0,1)\times[-1,2]$.
This proves all the induction assertions.

Finally, we prove the claimed termination bound.
Each completed cycle increases $s_i$ by more than $h/2$, while
$s_i\le T$.
There can therefore be at most $\lceil2T/h\rceil$ completed cycles.
The induction shows that every iteration labeled Run can be
completed, so the algorithm must eventually receive Stop or Project.
The returned point is an $\varepsilon$-approximate variational
inequality solution by Lemma \ref{asm: oracle}.
\end{proof}

\begin{proof}[of Theorem \ref{thm:optimization}.]
If the initial test returns $0$, it is correct because
$g(0)=\|F(0)\|\le\varepsilon/2$.
Otherwise, Lemma \ref{lem:combine} guarantees termination and correctness.
Each iteration uses a constant number of value and Jacobian queries,
so the total number is
\[
O\!\left(1+\frac{T}{h}\right)
=O\!\left(1+T\max\left\{\frac{L_J}{\theta},\frac{a}{\varepsilon}\right\}\right).
\]
The kernel computation and Newton solve use polynomially many
arithmetic operations in $n$ per iteration.
\end{proof}

%% file: assumption.tex
\subsection{High-Probability Regularity and Path-Length Bounds}\label{sec:regularization-perturbation}

\subsubsection{Preliminaries: Nets and Density Bounds}\label{sec: prelim of RMT}
We now prove the two probabilistic bounds needed by Theorem \ref{thm:optimization}: a lower bound on the smallest row singular value of $J_K$ along the zero set, and an upper bound on the length of the followed branch. We work directly with $K$, using fixed nets of the cylinder so that all density estimates remain unconditional.

Throughout this section, $n\ge2$, the bounds in Assumptions \ref{asm:F-bounded}, \ref{asm:F-Lip}, \ref{asm:F-smooth}, \ref{asm:bounded-density}, and \ref{asm:bounded-jacobian-density} are deterministic, and the regularity bounds hold almost surely. Recall
\[
a=L_0+L_1>0,\qquad
\widehat F=F/a,\qquad
K(x,t)=(1-t)x-t\widehat F(x).
\]

The density of $\widehat F(x)$ is at most $a^nP_v$, and the joint density of $(\widehat F(x),v^\top J_{\widehat F}(x))$ is at most $a^{2n}P_d$ for each fixed $x$ and unit $v$. Also,
\[
\|\widehat F\|\le L_0/a,\qquad
\|J_{\widehat F}\|\le L_1/a,\qquad
L_0/a+L_1/a=1.
\]
Thus $\|J_K\|\le L_K=3$ on the cylinder $\cB(0,1)\times[0,1]$. We retain the deterministic bound $L_J\ge2$ on the Lipschitz constant of $J_K$ from Section \ref{sec:optimization}.
Finally, we use the normalized density bounds
\[
V=\frac{a^nP_v}{n^{n/2}},
\qquad
D=\frac{a^{2n}P_d}{n^n}.
\]

Stirling's formula gives $\mathrm{Vol}(\cB_n)\le(\mathcal C/\sqrt n)^n$. Since $\widehat F(x)$ and $v^\top J_{\widehat F}(x)$ are supported in unit balls, their density bounds satisfy
\[
a^nP_v\,\mathrm{Vol}(\cB_n)\ge1,\qquad
a^{2n}P_d\,\mathrm{Vol}(\cB_n)^2\ge1.
\]
So $V,D\ge \mathcal C^{-n}$ for some universal constant $\mathcal C>1$. Since we use $\mathcal C$ to denote some absolute constant, we may use it to absorb the polynomial or exponential factors of $n$.

For an $n\times(n+1)$ matrix $A$, since the smallest singular value $\sigma_n(A)=\min_{\|v\|=1}\|A^\top v\|$, using triangle inequality gives that$|\sigma_n(A)-\sigma_n(B)|\le\|A-B\|$.
We use this inequality to transfer the singular-value estimate from net centers to nearby zeros.

Finally, we provide the following net estimates. 

\begin{lemma}\label{lem:net}
Let $0<\delta\le1$, and let $X$ be a maximal subset of $\cB(0,1)\subseteq\mathbb R^n$ such that no two points of $X$ are having distance $\le \delta$. $X$ is a $\delta$-net, and
\[
\bigl|\{x\in X:\|x\|\le R\}\bigr|
\le(1+2R/\delta)^n.
\]
It also satisfies
\[
\sum_{x\in X}\frac1{\max\{\|x\|,\delta\}}
\le \mathcal C^n\delta^{-n}.
\]
Similarly, in dimension $n\ge1$, the unit sphere admits an $\delta$-net of size at most $\mathcal C^n\delta^{-(n-1)}$ for $0<\delta\le1$, if we take a maximal set on $\mathbb S^{n-1}$ whose pairwise distance is greater than $\delta$.
\end{lemma}
\begin{proof}
It is obvious that $X$ is a $\delta$-net by maximality. Then, we can see that the balls of radius $\delta/2$ with centers of norm $\le R$ have disjoint interiors and lie in the ball of radius $R+\delta/2$. Comparing volumes gives the first bound.

For the weighted sum, the centers of norm at most $\delta$ contribute at most $3^n/\delta$. For each radius $R=2^j\delta<1$, the centers with $R<\|x\|\le\min\{2R,1\}$ contribute at most
\[
\frac1R(1+4R/\delta)^n
\le5^n\delta^{-n}R^{n-1}.
\]
Then, we take the sum through $R=\delta,2\delta,\dots$, which gives a geometric sequence, and the sum is dominated by the last term, which is $5^n/\delta^n$.

For the sphere, we take a maximal $\delta$-separated set. Similarly, for that set, the balls of radius $\delta/2$ and the centers of that set are inside the region between $\cB(0,1-\delta/2)$ and $\cB(0,1+\delta/2)$. Also, these balls are disjoint. Comparing the volumes of this shell and the balls gives at most $2n(2/\delta+1)^{n-1}$ points, which is bounded by $\cC^n\delta^{-n+1}$. 
\end{proof}

Using the same idea, we can prove the following lemma:
\begin{lemma}\label{lem:net-sphere}
    Let $0<\delta<1$ and let $n\ge2$. Suppose we have a sphere of radius $R\ge\delta$. Let $\delta\le h\le R$ and $M$ be the region on the sphere that $\{(x_1,\dots, x_n)|\sum_{i=1}^{n}x_i^2=R^2, |x_1|\le h\}$. Then we can have a $\delta$-net of $M$ with $\cC^n(1+R/\delta)^{n-2} h/\delta$ points using the construction of $\mathcal N\subset M$ the maximal set such that the distance between each pair of points is greater than $\delta$. In particular, if we choose to cover the whole sphere, the number of points needed is $\cC^n(1+R/\delta)^{n-1}$
\end{lemma}
\begin{proof}
It is easy to see that the construction is a $\delta$-net of $M$. To bound its cardinality, for each $p\in\mathcal N$ consider the spherical cap $C_p=\{q\in R\cdot\mathbb S^{n-1}:\|q-p\|<\delta/2\}$. These caps are pairwise disjoint. 
Write $v_{n-1}$ for the volume of the unit ball in $\mathbb R^{n-1}$. 
Therefore, each cap has surface area at least $v_{n-1}(\delta/4)^{n-1}$, because we can calculate that its orthogonal projection onto $p^\perp$ contains the ball
of radius $\delta/4$. 

We next bound the area occupied by these caps. We know that the union of the caps is contained in $\{(x_1,\dots, x_n)|\sum_{i=1}^{n}x_i^2=R^2, |x_1|\le h+\delta/2\}$. Let $b=\min\{1,(h+\delta/2)/R\}$. Therefore, the area is upper bounded by
\[
(n-1)v_{n-1}R^{n-1}
\int_{-b}^{b}(1-s^2)^{(n-3)/2}\,ds.
\]
For $n\ge3$, the integral is at most $2b$.
For $n=2$, it equals $2\arcsin b\le\pi b$, where the inequality
follows from convexity of $\arcsin$ on $[0,1]$.
Thus, in every dimension $n\ge2$, the area is at most
\[
\pi(n-1)v_{n-1}R^{n-1}b
\le \frac{3\pi}{2}(n-1)v_{n-1}R^{n-2}h,
\]
where we used $\delta\le h$.

Comparing this upper bound with the areas of the disjoint caps,
we obtain
\[
|\mathcal N|\le\frac{\frac{3\pi}{2}(n-1)v_{n-1}R^{n-2}h}
     {v_{n-1}(\delta/4)^{n-1}}\le
\cC^n\left(1+\frac R\delta\right)^{n-2}\frac h\delta.
\]
Finally, taking $h=R$ gives $M=R\cdot\mathbb S^{n-1}$ and yields
the stated bound for the whole sphere.
\end{proof}
\subsubsection{Bounding Small Values and Small Row Singular Values Together}\label{sec:proof condition}

Now we are going to see the regularity condition in Property \ref{asm:homotopy-gap}. For $t\le1/4$, the regularity can be seen since $\sigma_{\min}(\partial_xK)
\ge1-t(1+L_1/a)\ge\frac12.$
Appending the last column cannot decrease the smallest row singular value. We therefore need probabilistic estimates only for $t\in[1/4,1]$.

\begin{lemma}\label{lem:joint H}
For every fixed $(x,t)\in\cB(0,1)\times[1/4,1]$, $0<\eta\le1/4$, and $0<\alpha\le\eta$,
\[
\Pr\bigl[\|K(x,t)\|\le\alpha
\ \wedge\ \sigma_n(J_K(x,t))\le\eta\bigr]
\le \mathcal C^nD\,\frac{\alpha^n\eta^2}{\max\{\|x\|,\eta\}}.
\]
\end{lemma}
\begin{proof}
Write $h=K(x,t)$. $J_K(x,t)=[(1-t)I-tJ_{\widehat F}(x)~~~-x-\widehat F(x)]$. and $A=\partial_xK(x,t)=(1-t)I-tJ_{\widehat F}(x)$. On the event of $\|K(x,t)\|\le\alpha$ and $\sigma_n(J_K(x,t))\le\eta$, there is a unit vector $v$ with
\[
\|A^\top v\|\le\eta,\qquad
|\langle v,x+\widehat F(x)\rangle|\le\eta.
\]
Since $K(x,t)=(1-t)x-t\widehat F(x)$, we have $x=t(x+\widehat F(x))+h$, therefore by triangle inequality, we have
\[|\langle v,x\rangle|=|\langle v, t(x+\widehat F(x))+h\rangle|\le t|\langle v, (x+\widehat F(x))\rangle|+\|v\|\cdot \|h\|\le t\eta+\alpha\le2\eta.\]
This restricts $v$ to a deterministic band of the unit sphere.

Put $r=\|x\|$ and $\delta=\eta/4$. We claim that this band has a $\delta$-net $N$ with
\[
|N|\le \mathcal C^n\frac{\eta^{-(n-2)}}{\max\{r,\eta\}}.
\]

If $r\le4\eta$, by Lemma \ref{lem:net-sphere}, we just cover the whole sphere with $\cC^n\delta^{-(n-1)}\le\cC^n\eta^{-(n-2)}/\max\{r,\eta\}$. If $r>4\eta$, the band condition implies $|\langle v,\hat x\rangle|\le2\eta/r<1/2$. Therefore, by Lemma \ref{lem:net-sphere}, there exists a $\delta$-covering of
$\mathcal C^n\delta^{-(n-2)}\left(1+\frac{\eta}{r\delta}\right)\le \mathcal C^n\eta^{-(n-2)}/r$
points. 

We can check that the matrix $A$ has norm at most $1$. Therefore, there exists some $w\in N$ that satisfies $\|A^\top w\|\le \|A^\top v\|+\|A^{\top}(w-v)\|\le \eta+\delta\le2\eta$. Also, for each fixed $w$, the simultaneous conditions $\|h\|\le\alpha$ and $\|A^\top w\|\le2\eta$ restrict the joint vector $(\widehat F(x),w^\top J_{\widehat F}(x))$ to two balls of radii $\alpha/t$ and $2\eta/t$. Therefore, the joint density gives
\[
\Pr[\|h\|\le\alpha,\ \|A^\top w\|\le2\eta]\le a^{2n}P_d\,\mathrm{Vol}(\cB_n)^2(\alpha/t)^n(2\eta/t)^n
\le \mathcal C^nD\alpha^n\eta^n.
\]
Notice that the construction of the net $N$ depends only on the fixed $x,t,\alpha,\eta$. Therefore, we can use the union bound for the net and the probability above to prove this result.
\end{proof}

\subsubsection{Generalizing to All Zeros of \texorpdfstring{$K$}{K}}\label{sec: generalize}

We now apply the joint estimate to a fixed net of the cylinder. The factor involving $\|x\|$ is summable over its spatial centers by Lemma \ref{lem:net}, so no lower bound on $\|F(0)\|$ or on $\|x\|$ is needed.

\begin{lemma}\label{lem:delta}
There is a universal constant $\mathcal C>1$ such that, for every $0<q<1$, the choice
$\theta=\mathcal C^{-n}\frac{q}{D L_J}$
satisfies $\theta\le1/8$ and, with probability at least $1-q$, for all zeroes $(x,t)$ of $K$ in $\cB(0,1)\times[0,1]$, $\sigma_n(J_K(x,t))\ge\theta$. 
\end{lemma}
\begin{proof}
Since $D\ge \mathcal C^{-n}$ and $L_J\ge2$, a sufficiently large universal $\mathcal C$ ensures $\theta\le1/8$. Set
\[
\delta=\frac{\theta}{8L_J},\qquad
\alpha=5\delta,\qquad \eta=2\theta.
\]
Take the net $X$ of $\cB(0,1)$ from Lemma \ref{lem:net} and a $\delta$-net $I$ of $[1/4,1]$ with at most $1/\delta$ points. Their product covers this cylinder within distance $\sqrt2\delta$. If $K$ has a zero $z$ with row singular value less than $\theta$, then $t>1/4$ (otherwise $\sigma_n(J_K(x,t))\ge 1/2$). Let $z'$ be the point closest to $z$ to the joint net, we have
\[
\|K(z')\|\le 3\|z-z'\|\le3\sqrt2\delta<\alpha,~~~~~~
\sigma_n(J_K(z'))\le\theta+L_J\|z-z'\|\le\theta+\sqrt 2L_J\delta<\eta.
\]
Let $\mathcal E$ be the event that at least one $z$ such that $K(z)=0$ and $\sigma_n(J_K(z))<\theta$. So if $\mathcal E$ happens, there is a net point $z'$ such that $\|K(z')\|<\alpha$ and $\sigma_n(J_K(z'))<\eta$. By Lemma \ref{lem:joint H}, the union bound, and the second property at Lemma \ref{lem:net}, we have
\[
\Pr[\mathcal E]\le \mathcal C^nD\alpha^n\eta^2
\sum_{t\in I}\sum_{x\in X}
\frac1{\max\{\|x\|,\eta\}}\le \mathcal C^nD\alpha^n\eta^2\cdot 1/\delta\cdot\delta^{-n}
\le \mathcal C^nD L_J\theta\le q.
\]
This concludes the proof.
\end{proof}

\subsubsection{Path Length Control}\label{sec:path length}

We count the net cells that can contain zeros, and then bound the length within each such cell using the row singular value bound. 
\begin{lemma}\label{lem:length}
Fix a deterministic $0<\theta\le1/8$, and let
\[
\mathcal E=\left\{\exists z\in\cB(0,1)\times[0,1]:K(z)=0,\ \sigma_n(J_K(z))<\theta\right\}
\]
be the event in Lemma \ref{lem:delta}.
Let $T\in[0,\infty]$ denote the length of the outgoing branch,
followed until its first zero of $F$ or its first intersection
with the unit sphere.
On $\mathcal E^c$, this branch has finite length and reaches
one of these two terminal conditions.
Moreover, there is a universal constant $\mathcal C>1$ such that,
for every $0<q<1$,
\[
\Pr\left[(\neg\mathcal E)\wedge\left(1+T>\mathcal C^nV/q\right)\right]\le q.
\]
\end{lemma}

\begin{proof}
We first bound the length contributed by $0\le t\le1/4$.
For each such $t$, the equation $K(x,t)=0$ is equivalent to
\[x=\frac{t}{1-t}\widehat F(x).\]
The right-hand side maps the unit ball into itself and has Lipschitz constant at most $1/3$. Thus its unique fixed point
$x(t)$ satisfies $\|x(t)\|\le1/3$.
Also,
\[
\sigma_{\min}(\partial_xK)\ge1-t-t\|J_{\widehat F}\|\ge1-2t\ge\frac12.
\]
Differentiating $K(x(t),t)=0$ (here, we see $x$ as a function to $t$) therefore gives
$\partial_xK(x(t),t)x'(t)=x(t)+\widehat F(x(t))$,
And
$\|x'(t)\|\le2\left(\frac13+1\right)=\frac83.$
Therefore, this path $(x,t)$ for $t\le1/4$ has length at most
\[
\int_0^{1/4}\sqrt{1+\|x'(t)\|^2}\,dt\le\frac14\left(1+\frac83\right)=\frac{11}{12}<1.
\]
There is only one zero at $t=1/4$. After leaving this region, the simple outgoing branch cannot return to $t<1/4$ without
revisiting that zero. Hence the total contribution from $t\le1/4$ is less than $1$.

Now we consider the region $t\ge 1/4$. Consider the net we have constructed in the previous lemma. We dissect the region into different cells according to which point in the net that is closest to, as Voronoi diagram. The boundary can be seen as belonging to both of the cells We next count the net cells in the region st $t\ge 1/4$.
Set $\delta=\theta/(8L_J)$ and then we use the product
$N=X\times I$ from the proof of Lemma \ref{lem:delta}.
It has at most $\mathcal C^n\delta^{-(n+1)}$ points and covers $\cB(0,1)\times[1/4,1]$ within distance $\sqrt2\delta$.
Let $C_z$ be the cell associated with $z\in N$. If $C_z$ contains a zero $p$ of $K$, then $\|K(z)\|\le3\sqrt 2\delta$ (see proof in Lemma \ref{lem:delta}).
Thus the number of cells containing zeros is at most
$N_\delta:=\sum_{z\in N}\mathbf1_{\{\|K(z)\|\le5\delta\}}$.

We estimate $\mathbb E[N_\delta]$ directly from the value-density
bound. For a fixed center $z=(x,t)$, the event
$\|K(z)\|\le5\delta$ requires
\[
\widehat F(x)\in
\cB\left(\frac{1-t}{t}x,\frac{5\delta}{t}\right).
\]
Since $t\ge1/4$, the density bound for $\widehat F(x)$ gives
\[
\Pr[\|K(z)\|\le5\delta]\le a^nP_v\mathrm{Vol}(\cB_n)(5\delta/t)^n\le\mathcal C^n\frac{a^nP_v}{n^{n/2}}\delta^n=\mathcal C^nV\delta^n,
\]
where we used $\mathrm{Vol}(\cB_n)\le(\mathcal C/\sqrt n)^n$ by Stirling's formula.
Summing this over the net, we can get that the expected number of such cells is at most
\[
\mathbb E[N_\delta]\le\mathcal C^n\delta^{-(n+1)}V\delta^n
=\mathcal C^nV/\delta.
\]

We now bound the length inside each cell on $\neg\mathcal E$. Let $\Gamma$ be the branch remaining after $t=1/4$, with arclength parameterization $\gamma$, and for $z\in N$ let $S_z=\{s:\gamma(s)\in C_z\}$, so that the length of the path inside $C_z$ is the Lebesgue measure $|S_z|$. Put $r=\sqrt2\delta$, so that $C_z\subseteq\cB_{n+1}(z,r)$. Suppose $S_z\neq\emptyset$, fix $p_0=\gamma(s_0)$ with $s_0\in S_z$, write $J_0=J_K(p_0)$, and let $v$ be a unit vector spanning $\ker J_0$, which is one-dimensional since $\sigma_n(J_0)\ge\theta$ on $\neg\mathcal E$. Then $\|J_0w\|\ge\theta\|w\|$ for every $w\perp v$, and $\|J_K(p)-J_0\|\le2rL_J$ for every $p\in C_z$. Let $\phi(s)=\langle\gamma(s),v\rangle$.

First, $\phi$ is injective on $S_z$. Indeed, let $p,q\in\gamma(S_z)$ and write $d=q-p=\alpha v+w$ with $w\perp v$. Since $K(p)=K(q)=0$,
\[
\theta\|w\|\le\|J_0w\|=\|J_0d\|=\Bigl\|\int_0^1\bigl(J_0-J_K(p+\tau d)\bigr)d\,d\tau\Bigr\|\le2rL_J\|d\|=\frac{\sqrt2}{4}\,\theta\|d\|,
\]
so $\|w\|^2\le\|d\|^2/8$ and $|\phi(q)-\phi(p)|=|\alpha|\ge\sqrt{7/8}\,\|d\|$. Second, for $s\in S_z$ the unit tangent satisfies $J_K(\gamma(s))\gamma'(s)=0$, so the same computation with $d=\gamma'(s)$ gives $|\phi'(s)|=|\langle\gamma'(s),v\rangle|\ge\sqrt{7/8}$. Since $\phi$ is $1$-Lipschitz and injective on $S_z$, the area formula gives
\[
\sqrt{7/8}\,|S_z|\le\int_{S_z}|\phi'(s)|\,ds=|\phi(S_z)|\le2r,
\]
as $\phi(S_z)$ lies in an interval of length $2r$. Hence the length of the path inside each cell is at most $\sqrt{8/7}\cdot2\sqrt2\,\delta<4\delta\le\mathcal C^n\delta$. This argument does not assume that the length is finite.

Summing over the cells now gives a uniform bound for the path:
\[
\operatorname{length}(\Gamma)\le\mathcal C^n\delta N_\delta\qquad\text{on }\mathcal E^c.
\]
Since $N$ is finite, this proves $T<\infty$ on $\mathcal E^c$, and taking expectations gives $\mathbb E[\operatorname{length}(\Gamma)\mathbf1_{\mathcal E^c}]\le\mathcal C^n\delta\mathbb E[N_\delta]\le\mathcal C^nV$; Lemma \ref{asm:full rank} gives the stated terminal condition.
The constant length in the range of $t\le 1/4$ can be absorbed, thus we can just use $\cC^nV$ to bound the expected total length $\mathbb E[(1+T)\mathbf1_{\mathcal E^c}]$. Finally, we can use Markov's inequality to have
\[
\Pr\left[(\neg\mathcal E) \wedge\left(1+T>\mathcal C^nV/q\right)\right]\le\frac{q}{\mathcal C^nV}\mathbb E\bigl[(1+T)\mathbf1_{\mathcal E^c}\bigr]\le q.
\]
\end{proof}

%% file: smooth.tex
\subsection{Final Analysis}\label{sec:final-analysis}
In the smoothed case, we assume $F(x)$ is sampled from a distribution of functions that satisfies Assumptions \ref{asm:F-bounded}, \ref{asm:F-Lip}, \ref{asm:F-smooth}, \ref{asm:bounded-density}, and \ref{asm:bounded-jacobian-density}. In the worst-case setting (as a corollary of Theorem \ref{thm:smooth analysis final}), our algorithm perturbs $F(x)$ to $\tilde F=F(x)+Ax+b$, where $A,b$ are a matrix and a vector with i.i.d.~Gaussian entries. Therefore, the analysis in Sections \ref{sec:optimization} and \ref{sec:regularization-perturbation} and Algorithm \ref{alg:discrete}, applies to both settings.

\subsubsection{Smoothed Analysis}
We first formally present the theorem of runtime for smoothed analysis. Throughout this subsection, $n\ge2$.

\begin{theorem}[Smoothed-case algorithm]\label{thm:smooth analysis final}  
    For the Setting of Theorem~\ref{ithm:smoothed}, for all $0<p<1$, there exists a universal constant $\mathcal C$ such that there exists an algorithm that finds an $\varepsilon$-approximate solution to a variational inequality with probability at least $1-p$ in a runtime that is the maximum of the two following terms
    \[
        \mathcal C^n\frac{P_v}{n^{n/2}}(L_0+L_1)^{n+1}p^{-1}\varepsilon^{-1},
    \]

    and
    \[
        \mathcal C^n\frac{P_d}{n^n}\frac{P_v}{n^{n/2}}(L_0+L_1)^{3n-2}(L_0+L_1+L_2)^2p^{-2}.
    \]
    
\end{theorem}

\begin{proof}
In the smoothed analysis, the algorithm can fail if the homotopy has an ill-conditioned zero or if the followed path is too long. We apply Lemma \ref{lem:delta} with failure probability $p/2$, and also apply Lemma \ref{lem:length} when the singular value condition holds but the path-length bound fails with probability at most $p/2$. Thus, with probability at least $1-p$, both bounds hold.
We proceed to analyze the runtime of the algorithm on this event.

Now we gather all the pieces and use Theorem \ref{thm:optimization}. Let $a=L_0+L_1$. If $\|F(0)\|\le\varepsilon/2$, then we return $0$; otherwise, $a>\varepsilon/2$ and the algorithm uses the normalized field $\widehat F=F/a$ with accuracy $\varepsilon/a$. Write, as the proof in the last section, 
\[
D=\frac{a^{2n}P_d}{n^n},~~~~~~~~~
V=\frac{a^nP_v}{n^{n/2}},~~~~~~~~~L_J=\frac{2L_2}{a}+2\left(1+\frac{L_1}{a}\right).
\]
Previously in Section \ref{sec: prelim of RMT} we have showed that $D,V\ge \mathcal C^{-n}$. Thus, after enlarging the universal constant, the second final bound also absorbs the cost of an immediate return.
By Lemmas \ref{lem:delta} and \ref{lem:length}, we can choose $\theta$ as $\theta=\mathcal C^{-n}\frac{p}{D L_J}$.
So by Lemma \ref{lem:length} with probability at most $p/2$, the singular value condition holds and the length satisfies $T> \mathcal C^n\frac{V}{p}$.
The factors from using failure probability $p/2$ are absorbed into the universal constant $\mathcal C$.

We apply \ref{thm:optimization} to the homotopy $K(x,t)=(1-t)x-t\widehat F(x)$ and the number $N$ of the iterations satisfies
\[
N\le \mathcal C(1+T)\max\left\{\frac{L_J}{\theta},\frac{a}{\varepsilon}\right\}\le \mathcal C^n\max\left\{\frac{D V L_J^2}{p^2},\frac{aV}{p\varepsilon}\right\}.
\]

Substituting $D,V$ and using $L_J\le4(L_0+L_1+L_2)/a$, $a=L_0+L_1\le L_0+L_1+L_2$ gives
\[
\frac{aV}{p\varepsilon}=\frac{P_v}{n^{n/2}}a^{n+1}p^{-1}\varepsilon^{-1}\le \mathcal C^n\frac{P_v}{n^{n/2}}(L_0+L_1)^{n+1}p^{-1}\varepsilon^{-1}
\]
and
\[
\frac{D V L_J^2}{p^2}\le\mathcal C\frac{P_d}{n^n}\frac{P_v}{n^{n/2}}a^{3n-2}(L_0+L_1+L_2)^2p^{-2}\le \mathcal C^n\frac{P_d}{n^n}\frac{P_v}{n^{n/2}}(L_0+L_1)^{3n-2}(L_0+L_1+L_2)^2p^{-2}.
\]
\end{proof}

Now we instantiate this theorem as corollaries:

\begin{corollary}\label{cor:instantiate}
    Let $X$ and $Y$ be independent, density-bounded, sub-Gaussian, centered ($\mathbb E[X]=\mathbb E[Y]=0$) distributions. The densities of $X$ and $Y$ are at most $ P_X$ and $ P_Y$, respectively, and let $g_X$ and $g_Y$ be the sub-Gaussian norms of $X$ and $Y$. Let $F(x)=F_0(x)+Dx+e$ be the perturbed model, where $F_0$ is sampled from a pool of base functions that satisfy Assumptions \ref{asm:F-bounded}, \ref{asm:F-Lip}, and \ref{asm:F-smooth}. $D$ is a matrix with i.i.d. copies of $X/\sqrt{n}$ and $e$ is a vector with i.i.d. copies of $Y/\sqrt{n}$. The random objects $F_0,D,e$ are independent. Write $W=L_0+L_1+(g_X+g_Y)\sqrt{\log(4/p)}$. Then, with probability at least $1-p$, the algorithm finds an $\varepsilon$-approximate solution to the variational inequality for $F$ in a runtime bounded by the maximum of \[
        \mathcal C^n P_Y^n W^{n+1}p^{-1}\varepsilon^{-1},\]

    and
    \[\mathcal C^n(P_XP_Y)^n P_Y^n W^{3n-2}(W+L_2)^2p^{-2}.\]
    where $\mathcal C$ is a universal constant.
\end{corollary}
\begin{proof}
    For this corollary, there are two types of bad events: one is that the algorithm in Theorem \ref{thm:smooth analysis final} fails, and the other is that $\|D\|_{op}$ and $\|e\|$ are too large. We split the $p$ into two, $p/2$ is for $\|D\|_{op}$ and $\|e\|$ being out of range, and another $p/2$ is to handle the failure probability of Theorem \ref{thm:smooth analysis final} if $\|D\|_{op}$ and $\|e\|$ are within the range.

    Now we are going to prove two things: one is choosing new $L_0, L_1$ to adapt to the new noise so that the bounds on $L_0, L_1$ hold with high probability, and the other is the density bounds $P_v, P_d$.

    First, we can clearly bound that $P_v\le (P_Y\sqrt{n})^n$. Now we consider $P_d$. Fix $x$ and a unit vector $v$, and condition on some $F_0$. So we have, $v^\top J_F(x)=v^\top J_{F_0}(x)+v^\top D$. By Rudelson--Vershynin's density bound \cite[Theorem~1.2]{RudelsonVershynin2014Densities}, each coordinate of $v^\top D$ has density at most $\sqrt{2n}P_X$. Its coordinates are independent, since they use different columns of $D$, so its density is at most $(\sqrt{2n}P_X)^n$. The same bound holds for $Dv$, using the independence of the rows of $D$. Then, given $F_0,v^\top J_F(x)$, the density of $F(x)=F_0(x)+Dx+e$ is at most $(\sqrt nP_Y)^n$. So multiplying the density of $v^\top J_F(x)\mid F_0$ and the conditioning on $F_0,v^\top J_F(x)$ is $(\sqrt nP_Y)^n(\sqrt{2n}P_X)^n$. Then we marginalize all the $F_0$, The density of $(v^\top J_F(x), F(x))$ just averages these conditional densities, so the density is still upper bounded by $(\sqrt nP_Y)^n(\sqrt{2n}P_X)^n$.

    Then we choose a new set of $L_0$ and $L_1$. By \cite[Theorem~4.4.3]{Vershynin_2026} (applied to $D$, and to $e$ as a $1\times n$ matrix), there exists a universal constant $\mathcal C$ such that, for any $0<p<1$, with probability at least $1-p/2$, $\|D\|_{op}<\mathcal C g_X\sqrt{\log(4/p)}$ and $\|e\|\le\mathcal C g_Y\sqrt{\log(4/p)}$. Therefore, we can choose the new $L_0'$ to be $L_0+\mathcal C (g_X+g_Y)\sqrt{\log(4/p)}$ and the new $L_1'$ to be $L_1+\mathcal C g_X\sqrt{\log(4/p)}$. Condition on the event that both norm bounds hold. Its probability is at least $1-p/2\ge1/2$, so by Bayes rule, the conditioned densities of $F(x)$ and $(F(x), v^\top J_F(x))$ are bounded by $2P_v$ and $2P_d$, respectively. Thus, by taking the new $L_0', L_1',2P_v,2P_d$ and failure probability $p/2$ in Theorem \ref{thm:smooth analysis final}, we have the bound of the maximum of 
    \[\mathcal C^n P_Y^n W^{n+1}p^{-1}\varepsilon^{-1},\]

    and
    \[\mathcal C^n(P_XP_Y)^n P_Y^n W^{3n-2}(W+L_2)^2p^{-2}.\]
\end{proof}

Now we give another class of distributions that are from $\cB(0,1)\to\cB(0,1)$ which gives a fixed point.

\begin{corollary}\label{cor:instantiate2}
    Let $r_0\ge0$, $r_D,r_e>0$, and $r_0+r_D+r_e\le 1$. Consider $F(x)=F_0(x)+Dx+e$ where $F_0$ is sampled from a distribution of functions that satisfy \ref{asm:F-Lip}, \ref{asm:F-smooth} and $\|F_0(x)\|\le r_0<1$ for all $x\in\cB(0,1)$ (so $L_0=r_0$), and $D=rQ$, where $Q$ is a Haar-distributed random orthogonal matrix and $0\le r\le r_D$ where $r$ has density $nr^{n-1}/r_D^n$ on $[0,r_D]$ and furthermore $e$ is selected uniformly from $\cB(0,r_e)$. The random objects $F_0,Q,r,e$ are independent. Then, with probability at least $1-p$, an $\varepsilon$-approximate fixed point of $F$ can be found in a runtime bounded by the maximum of
    \[\mathcal C^n r_e^{-n}(1+L_1)^{n+1}p^{-1}\varepsilon^{-1},\]
    and
    \[\mathcal C^n(r_Dr_e)^{-n}r_e^{-n}(1+L_1)^{3n-2}(1+L_1+L_2)^2p^{-2}.\]
    where $\mathcal C$ is a universal constant.
\end{corollary}
\begin{proof}
    For this case, we can say that $P_v$ is, because of $e$, no more than $1/\mathrm{Vol}(\cB_n)\cdot r_e^{-n}$. Furthermore, $P_d$ is, because for all unit $v$, $v^{\top}D$ and $Dv$ are uniformly distributed on a ball of radius $r_D$, bounded by $1/\mathrm{Vol}(\cB_n)^2\cdot r_e^{-n}r_D^{-n}$. Here we use the Haar distribution of $Q$ and the chosen density of $r$, and the same conditioning argument as in Corollary \ref{cor:instantiate}.

    We have $\|F(x)\|\le r_0+r_D+r_e\le1$. Apply Theorem \ref{thm:smooth analysis final} to $F(x)-x$ with accuracy $\varepsilon$, bounds $2$, $L_1+r_D+1$, $L_2$. The shifts by $x$ and $I$ preserve the density bounds. In the ``Stop'' case of Lemma \ref{asm: oracle}, $\|F(x)-x\|\le\varepsilon/2$. In the ``Project'' case, its variational inequality condition, $\|F(\hat x)\|\le1$, and $\|\hat x\|=1$ give
    \[
    \|F(\hat x)-\hat x\|\le\varepsilon+\langle F(\hat x)-\hat x,\hat x\rangle
    \le\varepsilon+1-\|\hat x\|=\varepsilon.
    \]
    Thus, in the ``Project'' case, the returned point $x_1=x/\|x\|$ satisfies
    \[
    \|F(x_1)-x_1\|=\|F(\hat x)-\hat x\|\le\varepsilon.
    \]
    Since $r_D\le1$, the sum of the norm and Lipschitz bounds used above satisfies $L_1+r_D+3\le4(1+L_1)$, and adding $L_2$ gives at most $4(1+L_1+L_2)$. Plugging these bounds into Theorem \ref{thm:smooth analysis final}, using $1/(\mathrm{Vol}(\cB_n)n^{n/2})\le\mathcal C^n$, we can get the result.
\end{proof}

\subsubsection{Worst-Case Proof of the Algorithm}
We can also use Corollary \ref{cor:instantiate} to analyze our algorithm in the worst case. The proof uses Corollary \ref{cor:instantiate} with some manually added Gaussian perturbation in order to find an approximate solution to the variational inequality.

\begin{theorem}[Worst-case algorithm]\label{thm:algorithm final}
    For a function $F$ that satisfies Assumptions \ref{asm:F-bounded}, \ref{asm:F-Lip} and \ref{asm:F-smooth}, for all $n\ge2$, $0<\varepsilon\le1$ and $0<p<1$, there exists a universal constant $\mathcal C$, such that there exists an algorithm that finds an $\varepsilon$-approximate solution to a variational inequality with probability at least $1-p$ in runtime
    \[\mathcal C^n\varepsilon^{-3n}\cdot\log(4/p)^{1.5n}\cdot(1+L_0+L_1)^{3n-2}(1+L_0+L_1+L_2)^2p^{-2}.\]
\end{theorem}
\begin{proof}
    We consider adding a noise to $F$. Consider $\tilde F(x)=F(x)+Dx+e$ where $D$ is a matrix of i.i.d. $\mathcal N(0,\varepsilon^2/(\mathcal C n\log(4/p)))$ entries and $e$ is a vector of i.i.d. $\mathcal N(0,\varepsilon^2/(\mathcal C n\log(4/p)))$ entries, independently, for some sufficiently large absolute constant $\mathcal C$. Therefore we know that with probability at least $1-p/2$, $\|D\|_{op}\le\varepsilon/8$ and $\|e\|\le\varepsilon/8$ and thus $\|\tilde F(x)-F(x)\|\le\varepsilon/4$ uniformly on $\cB(0,1)$. Therefore, we have, for the returned point and every $y\in\cB(0,1)$,
    \[\langle F(x),y-x\rangle\le\langle\tilde F(x),y-x\rangle+\|\tilde F(x)-F(x)\|\,\|y-x\|\le\varepsilon/2+2(\varepsilon/4)=\varepsilon.\]
    
    So the returned point is a valid $\varepsilon$-approximate solution for the variational inequality. For the time bound, using Corollary \ref{cor:instantiate} with variational inequality accuracy $\varepsilon/2$ and failure probability $p/2$, we can take $P_X=P_Y=\mathcal C\sqrt{\log(4/p)}/\varepsilon$, and upper bound $g_X$ and $g_Y$ by $\mathcal C\varepsilon/\sqrt{\log(4/p)}$. Plugging them into the bounds of Corollary \ref{cor:instantiate}, we get the maximum of
    \[
    \mathcal C^n\varepsilon^{-(n+1)}\log(4/p)^{n/2}
    (1+L_0+L_1)^{n+1}p^{-1}
    \]
    and 
    \[\mathcal C^n\varepsilon^{-3n}\cdot\log(4/p)^{1.5n}\cdot(1+L_0+L_1)^{3n-2}(1+L_0+L_1+L_2)^2p^{-2}.\] 
    The latter dominates for $n\ge2$, $0<\varepsilon\le1$, and $0<p<1$. With probability at least $1-p$, both the norm event and the algorithmic guarantee hold. 
\end{proof}

%% file: lowerBound.tex
\section{Lower Bounds} \label{sec:lowerBound}

In this section we will prove a lower bound on the query complexity of finding fixed points of continuous maps $F : K \to K$ when the diameter of $K$ is constant. The lower bound follows from the construction in \cite{rubinstein2017settling}, which uses relative norms, together with a scaling argument. We give this argument below instead of repeating the construction.

\begin{theorem} \label{thm:lowerBoundMain}
  There are absolute constants $L_1,\varepsilon_0>0$ such that, in the worst case, $2^{\Omega(n)}$ value queries are needed to find an $\varepsilon_0$-approximate fixed point of an $L_1$-Lipschitz function $F : \mathcal{B}_n \to \mathcal{B}_n$.
 This implies a $2^{\Omega(n)}$ value-query lower bound for finding $\varepsilon_0$-approximate solutions of the variational inequality problem over $\mathcal{B}_n$, even under Assumptions~\ref{asm:bounded-density}, \ref{asm:F-bounded}, and~\ref{asm:F-Lip} (the setting of Theorem~\ref{thm:smoothed}).
\end{theorem}

We present the proof of this theorem in the following section.

\subsection{Proof of Theorem \ref{thm:lowerBoundMain}} \label{sec:lowerBoundMain:proof}

\begin{proof}
Rubinstein's construction in \cite[Theorem~3.2 and Corollary~4.4 of the full version]{rubinstein2017settling} gives a value-query lower bound of $2^{\Omega(n)}$ for an $O(1)$-Lipschitz map $f:[0,1]^n\to[0,1]^n$ at constant accuracy in the relative norm $\|z\|_{2,\mathrm{rel}}=\|z\|_2/\sqrt n$. Thus, for some absolute constant $\rho>0$, finding $x$ with $\|f(x)-x\|_2\le\rho\sqrt n$ requires $2^{\Omega(n)}$ queries. Using the relative norm on both sides of the Lipschitz inequality leaves the Lipschitz constant unchanged.

Let $c=\tfrac12\mathbf 1$ and scale both the domain and range by $1/\sqrt n$:
\[
Q_n=\left[-\frac{1}{2\sqrt n},\frac{1}{2\sqrt n}\right]^n
\subseteq\mathcal B_n(0,1/2),\qquad
g(u)=\frac{f(c+\sqrt n\,u)-c}{\sqrt n}.
\]
Then $g:Q_n\to Q_n$ has the same Lipschitz constant as $f$, and
\[
\|g(u)-u\|_2
=\frac{\|f(c+\sqrt n\,u)-(c+\sqrt n\,u)\|_2}{\sqrt n}.
\]
To obtain a self-map of the entire unit ball, define $F(z)=g(\Pi_{Q_n}z)$, where $\Pi_{Q_n}$ is Euclidean projection onto $Q_n$. Since projection is non-expansive, $F$ is $O(1)$-Lipschitz and maps $\mathcal B_n$ into $Q_n$. For $w=\Pi_{Q_n}z$, we have
\[
\|g(w)-w\|_2
=\|\Pi_{Q_n}F(z)-\Pi_{Q_n}z\|_2
\le\|F(z)-z\|_2.
\]
Hence a $\rho$-approximate fixed point of $F$ gives a $\rho\sqrt n$-approximate fixed point of $f$. Each query to $F$ requires only one query to $f$, so the value-query lower bound is preserved.

For the smoothed VI claim, let $b$ be uniform on $\mathcal B_n(0,\sigma)$ for a sufficiently small absolute constant $0<\sigma<1/2$, and set $V_b(x)=F(x)-x+b$. This is a bounded Lipschitz vector field, and for every fixed $x$, $V_b(x)$ has density at most $1/\mathrm{Vol}(\mathcal B_n(0,\sigma))$. If $x$ is an $\eta^2$-approximate VI solution, taking $y=F(x)$ and writing $d=F(x)-x$ gives
\[
\|d\|_2^2\le\eta^2+\sigma\|d\|_2,
\qquad\text{hence}\qquad
\|d\|_2\le\eta+\sigma.
\]
Choosing $\eta+\sigma\le\rho$ transfers the same lower bound to constant-accuracy solutions of the smoothed VI. Each query to $V_b$ uses one query to $F$. Taking $b=0$ also gives the worst-case VI lower bound.
The same perturbation also gives the smoothed fixed-point lower bound: $F+b$ maps $\mathcal B_n$ to itself, has bounded value densities, and $\|F(x)-x\|_2\le\|F(x)+b-x\|_2+\sigma$.
\end{proof}